\documentclass[aps,pra,reprint,amsmath,amssymb,superscriptaddress,nofootinbib]{revtex4-2}
\usepackage{amsmath, amssymb, mathrsfs, braket}
\usepackage{hyperref}
\usepackage{amsthm}
\usepackage{xcolor}
\usepackage{mathtools}
\usepackage{yquant}
\usepackage{tikz-cd}

\newtheorem{theorem}{Theorem}
\newtheorem{lemma}[theorem]{Lemma}
\newtheorem{proposition}[theorem]{Proposition}

\DeclareMathOperator{\sym}{Sym}
\DeclareMathOperator{\cnot}{CNOT}
\DeclareMathOperator{\swapop}{SWAP}
\DeclareMathOperator{\tr}{Tr}
\DeclareMathOperator{\rank}{rank}

\DeclareMathOperator{\imag}{Im}
\DeclareMathOperator{\qft}{QFT}

\begin{document}

\title{Symmetry-Based Quantum Codes Beyond the Pauli Group}

\author{Zachary P. Bradshaw}
\thanks{Corresponding Author: zak@qodexquantum.ai}
\affiliation{QodeX Quantum, Chicago, IL}
\author{Margarite L. LaBorde}
\affiliation{QodeX Quantum, Chicago, IL}
\author{Dillon Montero}
\affiliation{Tulane University, Department of Mathematics, New Orleans, LA}
\date{\today}

\begin{abstract}
Typical stabilizer codes aim to solve the general problem of fault-tolerance without regard for the structure of a specific system. By incorporating a broader representation-theoretic perspective, we provide a generalized framework that allows the code designer to take this structure into account. For any representation of a finite group, we produce a quantum code with a code space invariant under the group action, providing passive error mitigation against errors belonging to the image of the representation. Furthermore, errors outside this scope are detected and diagnosed by performing a projective measurement onto the isotypic components corresponding to irreducible representations of the chosen group, effectively generalizing syndrome extraction to symmetry-resolved quantum measurements. We show that all stabilizer codes are a special case of this construction, including qudit stabilizer codes, and show that there is a natural one logical qubit code associated to the dihedral group. Thus we provide a unifying framework for existing codes while simultaneously facilitating symmetry-aware codes tailored to specific systems.
\end{abstract}

\maketitle

\section{Introduction}

Quantum error correction (QEC) \cite{roffe2019,bradshaw2025,devitt2013,knill2000,steane1996,kribs2005,zanardi1997,girvin2023} plays a pivotal role in the realization of fault-tolerant quantum computing \cite{gottesman2009,shor1996,preskill1997,bacon2000}. Indeed, quantum information \cite{nielsen2000,watrous2018,petz2007} is notoriously fragile. To store and manipulate it reliably, the system must be encoded so that errors can be detected and corrected without disturbing the underlying logical state. QEC addresses this challenge by embedding qubits into larger Hilbert spaces so that information can be spread out, creating a form of redundancy without violating the no-cloning theorem \cite{calderbank1996,shor1995}. In some versions of QEC, the subsystem in which the logical state lives is invariant under the action of a significant portion of the noise channel, providing a passive form of protection. Other variants include an active approach that employs measurements to diagnose errors while preserving the coherence of the logical information. A recovery operation is then applied in an attempt to restore the corrupted state to its original form \cite{edmonds1965,fowler2009,fowler2014,poulin2008,duclos-cianci2010,davaasuren2020,gicev2023,baireuther2018,varsamopoulos2017,varsamopoulos2020,sweke2020,matekole2022,Domingo-Colomer2020}.

The most well-known family of codes are stabilizer codes~\cite{gottesman1997,bradshaw2025}, which define a subspace stabilized by an abelian subgroup of the Pauli group. When an error occurs, a measurement of the generators of this subgroup can be performed such that the outcomes provide information about the type of error. The success of this model has led to many fault-tolerant quantum computing architectures using special stabilizer codes such as surface codes~\cite{dennis2002,kitaev2003,fowler2012}, color codes~\cite{bombin2006,kubica2015,fowler2011}, and low-density parity-check codes \cite{breuckmann2021,panteleev2021a,babar2015,dinur2023,panteleev2021b}. Despite their effectiveness, stabilizer codes are intrinsically tied to abelian symmetry. Each stabilizer corresponds to a symmetry generator, and all such generators commute. It is natural to ask whether this framework can be generalized to incorporate nonabelian group symmetries, leading to a richer theory that leverages representation theory \cite{fulton2013,steinberg2011,sengupta2011}. Subsystem codes \cite{aly2006} take a step in this direction by introducing a nonabelian gauge group whose center is the stabilizer subgroup. The code space defined by this stabilizer splits into a tensor product of a subsystem, where the gauge group acts trivially, and a co-subsystem, where the gauge group acts nontrivially. Logical states are then encoded in the subsystem, which ensures that the state is inherently protected from errors in the gauge group. When the co-subsystem is trivial, the subsystem becomes a decoherence-free subspace \cite{lidar1998}.

In this paper, we propose an alternative perspective. Our starting point is the observation that a finite group $G$ acting on a Hilbert space $\mathcal{H}$ via a unitary representation $W$ partitions the space into invariant components associated with irreducible representations of $G$ \cite{steinberg2011,sengupta2011}. The code space is taken to be the $G$-symmetric subspace consisting of all $\ket{\psi}\in\mathcal{H}$ that are left invariant by the action of $G$.  Errors that preserve the group symmetry leave the state in this subspace, while symmetry-breaking errors push the state into one or more isotypic subspaces of a non-trivial component irrep. By performing a modification of a $G$-Bose symmetry test \cite{laBorde2023,laborde2023thesis,bradshaw2023} with respect to the irreducible components of $W$, we can diagnose the nature of the error. This perspective therefore generalizes syndrome extraction to the broader context of symmetry-resolved quantum measurements, similar to approaches taken to quantify noise in symmetric quantum dynamics and error-correcting codes~\cite{watkins2025}.

Our code can be understood as a decoherence-free subspace in which a potentially non-abelian group acts trivially, with the caveat that there is also a well-defined syndrome extraction procedure. The decoherence-free subspace provides passive protection against errors given by a linear combination of group elements, and the isotypic syndrome extraction procedure provides active protection against errors outside the span of the group algebra.

The usual stabilizer codes fit neatly into this perspective. Indeed, given a stabilizer subgroup, the code space is taken to be the corresponding symmetric subspace \cite{harrow2013}. Since this subgroup of the Pauli group is taken to be abelian and each element squares to the identity, its irreducible representations are of the form $\rho_j(k)=(-1)^{j\cdot k}$, where $j,k\in\mathbb{Z}_2^{\oplus n}$, and the usual syndrome extraction procedure produces the label corresponding to the isotypic subspace of each irrep. This construction generalizes easily to the case of qudit stabilizer codes \cite{rains1999,ashikhmin2001,grassl2003,gheorghiu2014}.

Several symmetry-based codes have appeared in specialized settings \cite{gross2021, denys2024, wang2025, kubischta2023, kubischta2024}. Our broader perspective highlights that such codes are not isolated phenomena but rather a special case of a unifying representation-theoretic framework for quantum error correction. It also suggests that nonabelian symmetries can be incorporated in a natural way that differs from the subsystem codes. Thus, we provide a mathematically precise formalism for the creation of quantum codes that go beyond the popular stabilizer formalism, and as an example, we give a natural single qubit code associated to the dihedral group. While preparing this manuscript we became aware of an independent work that unifies several distinct codes \cite{kubischta2025}, showing that for any choice of representation, there is a single intrinsic code which determines the error correcting properties of any of its physical realizations.

The remainder of this work is organized as follows. In Section~\ref{sec:sym-codes}, we introduce our representation theoretic perspective. We define the code space to be the symmetric subspace of a chosen group, show that syndrome extraction generalizes to this higher level perspective, comment on the form of logical operations, and provide a loose rate estimate in the most general setting. We also provide a mechanism for implementing two qubit gates between independent single qubit codes using the CNOT gate as an example. In Section~\ref{sec:examples}, we produce several code examples that showcase the power and generality of our construction. We show that all qudit stabilizer codes are $\mathbb{Z}_d^{\oplus n}$ codes with a natural choice of representation and give several examples which cannot be viewed as stabilizer codes in the usual sense. Among these examples is a single logical qubit code given by the basis permutation representation of the dihedral group. In Section~\ref{sec:qft}, we elaborate on this example, providing circuits for the dihedral group operations. In Section~\ref{sec:conclusion}, we give concluding remarks.

\section{Symmetry-Based Quantum Codes}
\label{sec:sym-codes}

Let $W:G\to\mathcal{U}(\mathcal                                                                                                                                                                                                                                 {H})$ be a unitary representation of some group $G$, which is not necessarily abelian or a subgroup of the Pauli group. The $G$-symmetric subspace
\begin{equation}
    \mathrm{Sym}_G := \{ \ket{\psi} \in \mathcal{H} : W(g)\ket{\psi} = \ket{\psi} \ \forall g \in G \}
\end{equation}
is the space of all states which are invariant under the action of $G$. The projection onto this subspace is an average over the elements of the group \cite{laBorde2023,harrow2013}.
\begin{lemma}\label{lemma:symg-projection}
    The projection $\Pi_G$ onto $\sym_G$ is given by
    \begin{equation}\label{eq:symg-projection}
        \Pi_G=\frac{1}{|G|}\sum_{g\in G}W(g).
    \end{equation}
\end{lemma}
\begin{proof}
    Observe that
    \begin{align}
        \Pi_G^2 &= \frac{1}{|G|^2}\sum_{g,h\in G}W(g)W(h)\\
        &=\frac{1}{|G|^2}\sum_{g,h\in G}W(gh)\\
        &=\frac{1}{|G|}\sum_{k\in G}W(k)=\Pi_G,
    \end{align}
    establishing that $\Pi_G$ is indeed a projection. To see that this projection is onto $\sym_G$, let $\ket{\psi}\in\imag(\Pi_G)$. Then there is a $\ket{\phi}$ such that $\ket{\psi}=\Pi_G\ket{\phi}$. It follows that
    \begin{align}
        W(h)\ket{\psi}&=W(h)\Pi_G\ket{\phi}\\
        &=\frac{1}{|G|}\sum_{g\in G}W(h)W(g)\ket{\phi}\\
        &=\frac{1}{|G|}\sum_{g\in G}W(g)\ket{\phi}=\ket{\psi},
    \end{align}
    for all $h\in G$. Thus, $\ket{\psi}\in\sym_G$. Conversely, if $\ket{\psi}\in\sym_G$, then $\Pi_G\ket{\psi}=\ket{\psi}$, and so it follows that the image of $\Pi_G$ is precisely $\sym_G$.
\end{proof}

We assume that $G$ is a finite group so that $W$ is completely reducible. Thus, $W$ decomposes as
\begin{equation}
    W = \bigoplus_{\lambda\in\widehat{G}}\rho_\lambda^{\oplus m_\lambda},
\end{equation}
where $\Hat{G}$ labels the equivalence classes of all irreducible representations of $G$ and $m_\lambda$ is the multiplicity of the irreducible representation $\rho_\lambda$. On the Hilbert space level, we have the decomposition
\begin{equation}\label{eq:hilbert-decomp}
    \mathcal{H}=\bigoplus_{\lambda\in\widehat{G}}\mathcal{H}_\lambda^{\oplus m_\lambda},
\end{equation}
where each $\mathcal{H}_\lambda$ is the invariant subspace corresponding to $\rho_\lambda$ and each $\mathcal{H}_\lambda^{\oplus m_\lambda}$ is referred to as the isotypic component associated to the irrep $\rho_\lambda$. Note that the isotypic component corresponding to the trivial representation is exactly the $G$-symmetric subspace. Indeed, this isotypic component is by definition the subspace of $\mathcal{H}$ for which $W$ acts as the trivial representation $W(g)=1$ for all $g\in G$.

\subsection{The Code Space}
The code space induced by the choice of $G$ is $\sym_G$, allowing us to detect any errors that rotate the logical state outside of $\sym_G$ by measuring the orthogonal projector \eqref{eq:symg-projection} onto this subspace. This procedure can be accomplished by performing a $G$-Bose symmetry test as in \cite{harrow2013,laBorde2023}. The reader is referred to \cite{bradshaw2023} for circuits that measure the symmetric and cyclic group projectors. Note that, just as in the usual stabilizer codes, the inclusion of $-I$ in the image of $G$ is unwise since it forces a trivial code space. This observation is the reason that one cannot study a nonabelian stabilizer subgroup when its elements are restricted to the Pauli group. Any two elements of the Pauli group either commute or anticommute, and so the inclusion of two elements which do not commute forces the group to contain $-I$ in order to be closed. Thus, any nonabelian subgroup of the Pauli group has a trivial code space. We will therefore have to extend beyond the Pauli group to consider such symmetries.

Let us make our construction for detecting an error explicit. We consider a logical codeword $\ket{\psi}_L \in \mathrm{Sym}_G$, and assume a unitary error operator $E$ acts on this state. In general, $E$ may not commute with $W(g)$, thus the corrupted state $E \ket{\psi}_L$ may no longer be symmetric. To detect this, we apply a $G$-Bose symmetry test (shown in Figure~\ref{fig:g-bose}) that measures the overlap with the code space by computing
\begin{equation}
    p = \bra{\psi} E^\dagger \Pi_G E \ket{\psi} = \| \Pi_G E \ket{\psi} \|^2.
\end{equation}
If $p < 1$, we conclude that the state has partially leaked outside the code space, indicating a detectable error. If $p=0$, the state is completely orthogonal to the code space, and error detection is conclusive. Thus, the projector $\Pi_G$ serves as a generalized stabilizer measurement detecting symmetry-breaking errors. The downside of this approach is that after the measurement, a corrupted state is potentially altered if the error that occurred does not commute with the projector $\Pi_G$. %unless the error that occurred happens to commute with the projection $\Pi_G$. 
Ideally, we would make a measurement that determines exactly which subspace the logical state has rotated into so that we can attempt to correct the error. This is the subject of the next section.

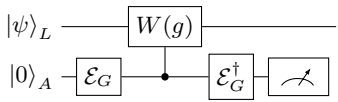
\begin{figure}
    \begin{tikzpicture}
        \begin{yquant}
    
            qubit {$\ket{\psi}_L$} L;
            qubit {$\ket{0}_A$} A;
            hspace {0.1cm} A;
            box {$\mathcal{E}_G$} A;
            box {$W(g)$} L | A;
            box {$\mathcal{E}_G^\dagger$} A;
            hspace {0.1cm} A;
            measure A;
            discard A;
    
        \end{yquant}
    \end{tikzpicture}
    \caption{$G$-Bose Symmetry Test. The gate $\mathcal{E}_G$ prepares the uniform superposition over group algebra basis for $G$ and the controlled $W(g)$ operation triggers when the ancillary register is in the state $\ket{g}$.}
    \label{fig:g-bose}
\end{figure}

\subsection{Isotypic Syndrome Extraction}
\label{sec:syndrome_extraction}

To fully diagnose the error, beyond detection, we must determine into which symmetry sector the corrupted state was mapped. We do so by measuring projection operators onto isotypic components. The Hilbert space decomposes as \eqref{eq:hilbert-decomp} and the projectors onto each isotypic component $\mathcal{H}_\lambda^{\oplus m_\lambda}$ are given in the following lemma.
\begin{lemma}\label{lemma:lambda-projection}
    Let $d_\lambda$ be the dimension of the irrep $\rho_\lambda$. The operators
    \begin{equation}
        \Pi_\lambda=\frac{d_\lambda}{|G|}\sum_{g\in G}\chi_\lambda(g^{-1})W(g)
    \end{equation}
    are the orthogonal projections onto each isotypic component $\mathcal{H}^{\oplus m_\lambda}$.
\end{lemma}
\begin{proof}
    We first show that each $\Pi_\lambda$ is a projection. Observe that
    \begin{align}
        \Pi_\lambda^2&=\frac{d_\lambda^2}{|G|^2}\sum_{g,h\in G}\chi_\lambda(g^{-1})\chi_\lambda(h^{-1})W(gh)\\
        &=\frac{d_\lambda^2}{|G|^2}\sum_{h,k\in G}\chi_\lambda(hk^{-1})\chi_\lambda(h^{-1})W(k).\label{eq:sub-here}
    \end{align}
    Now notice that
    \begin{align}
        \chi_\lambda(hk^{-1})\chi_\lambda(h^{-1})&=\tr[\rho_\lambda(hk^{-1})\otimes\rho_\lambda(h^{-1})]\\
        &=\tr[(\rho_\lambda(h)\otimes\rho_\lambda(h^{-1}))(\rho_\lambda(k^{-1})\otimes I)],
    \end{align}
    so that summing over $h$ produces
    \begin{equation}
    \begin{aligned}
        \sum_{h\in G}&\chi_\lambda(hk^{-1})\chi_\lambda(h^{-1})\\
        &=\tr\left[\sum_{h\in G}(\rho_\lambda(h)\otimes\rho_\lambda(h^{-1}))(\rho_\lambda(k^{-1})\otimes I)\right],
    \end{aligned}
    \end{equation}
    but by the matrix element Schur orthogonality relations, 
    \begin{equation*}
        \sum_{g\in G}\rho_\lambda(g)^*_{ij}\rho_\mu(g)_{kl}=\delta_{\lambda\mu}\delta_{ik}\delta_{jl}\frac{|G|}{d_\lambda} \, .
    \end{equation*} From here, it can be shown by direct calculation that $\sum_{h\in G}(\rho_\lambda(h)\otimes\rho_\lambda(h^{-1}))=\frac{|G|}{d_\lambda}\swapop$. Thus,
    \begin{equation}
        \sum_{h\in G}\chi_\lambda(hk^{-1})\chi_\lambda(h^{-1})=\frac{|G|}{d_\lambda}\tr\left[\swapop(\rho_\lambda(k^{-1})\otimes I)\right],
    \end{equation}
    and since $\tr\left[\swapop(A\otimes I)\right]=\sum_{i,j}\bra{ij}\swapop(A\otimes I)\ket{ij}=\sum_{i,j}\bra{ji}(A\otimes I)\ket{ij}=\tr[A]$, we have
    \begin{equation}
        \sum_{h\in G}\chi_\lambda(hk^{-1})\chi_\lambda(h^{-1})=\frac{|G|}{d_\lambda}\chi_\lambda(k^{-1}).
    \end{equation}
    Now substituting into \eqref{eq:sub-here}, we see that
    \begin{equation}
        \Pi_\lambda^2 = \frac{d_\lambda}{|G|}\sum_{k\in G}\chi_\lambda(k^{-1})W(k)=\Pi_\lambda,
    \end{equation}
    establishing that $\Pi_\lambda$ is a projection. To see that this is an orthogonal projection, let $\lambda,\mu\in\widehat{G}$ with $\lambda\ne\mu$. Then for any $\ket{\psi_\mu}\in\mathcal{H}_\mu^{\oplus m_\mu}$, we have
    \begin{align}
        \Pi_\lambda\ket{\psi_\mu}&=\frac{d_\lambda}{|G|}\sum_{g\in G}\chi_\lambda(g^{-1})W(g)\ket{\psi_\mu}\\
        &=\frac{d_\lambda}{|G|}\sum_{g\in G}\chi_\lambda(g^{-1})\rho_\mu^{\oplus m_\mu}(g)\ket{\psi_\mu},
    \end{align}
    and this latter sum vanishes by the Schur orthogonality relations. This also shows that the image of $\Pi_\lambda$ is contained in $\mathcal{H}_\lambda^{\oplus m_\lambda}$. All that remains is to show that $\mathcal{H}_\lambda^{\oplus m_\lambda}$ is contained in the image of $\Pi_\lambda$. To this end, let $\ket{\psi_\lambda}\in\mathcal{H}_\lambda^{\oplus m_\lambda}$ and expand $\ket{\psi_\lambda}=\sum_{k}\psi_k\ket{k}$. Then
    \begin{align}
        \Pi_\lambda\ket{\psi_\lambda} &= \frac{d_\lambda}{|G|}\sum_{g\in G}\sum_{i,j,k}\chi_\lambda(g^{-1})[\rho^{\oplus m_\lambda}_\lambda(g)]_{i,j}\psi_k\ket{k}\\
        &=\sum_{k}\psi_k\ket{k}=\ket{\psi_\lambda},
    \end{align}
    where we have again applied Schur orthogonality of matrix elements.
\end{proof}

Performing a measurement with respect to these projectors resolves the corrupted state into irreducible sectors. We take as our syndrome the label $\lambda$ of the measured irrep and perform a recovery operation based on this label. In contrast to the usual stabilizer syndrome extraction perspective, we view syndrome extraction as measuring the symmetry content of the state rather than eigenvalues of commuting Pauli operators. Indeed, it turns out that stabilizer codes fit nicely into this symmetry-aware picture. Measuring these eigenvalues can be viewed as measuring symmetry with respect to the action of a product of cyclic groups (see Section~\ref{sec:stabilizer-codes} for more details).

Let us make our analog of syndrome extraction more precise. We proceed with a structure similar to a $G$-Bose symmetry test \cite{harrow2013,laBorde2023,bradshaw2023,laborde2023thesis} but modify the ancillary register slightly by applying a quantum Fourier transform (QFT) over $G$ in the final step, as shown in Figure~\ref{fig:isotypic-syndrome}. Given a corrupted state $E\ket{\psi}_L$, we first append an ancillary register in the all-zeros state $\ket{0}_A$ and transform it into a uniform superposition over the basis $\{\ket{g}_A\}_{g\in G}$ of the group algebra $\mathbb{C}[G]$, leaving the state
\begin{equation}\label{eq:ancilla+input1}
    \frac{1}{\sqrt{|G|}}\sum_{g\in G}E\ket{\psi}_L\ket{g}_A.
\end{equation}
The next step %in a $G$-Bose symmetry test 
is to apply the operation $W(g)$ to the logical register conditioned on the ancillary register being in the $\ket{g}$ state. This produces the state
\begin{equation}\label{eq:ancilla+input2}
    \frac{1}{\sqrt{|G|}}\sum_{g\in G}W(g)E\ket{\psi}_L\ket{g}_A.
\end{equation}
In an ordinary $G$-Bose symmetry test, we would now uncompute the ancillary register by applying the adjoint of the operator that produces the uniform superposition and conclude by measuring in the computational basis. Instead, we perform a QFT over $G$ on the ancillary register, denoted $\operatorname{QFT}_G$ \cite{moore2006,kawano2016,beals1997,hoyer1997}. For groups like $S_3$, explicit circuits implementing the $\operatorname{QFT}_G$ are known \cite{kawano2016}. For larger or more complex groups, one may exploit group structure or approximate QFT techniques. 

The $\operatorname{QFT}_G$ maps the basis of the group algebra to the Fourier basis indexed by irreducible representations and representation spaces:
\begin{equation}\label{eq:group-QFT}
    \textnormal{QFT}_G : \ket{g} \mapsto \sum_{\lambda \in \widehat{G}} \sqrt{\frac{d_\lambda}{|G|}} \sum_{i,j=1}^{d_\lambda} [\rho_\lambda(g)]^*_{i,j} \ket{\lambda, i, j},
\end{equation}
where $d_\lambda$ is the dimension of the irrep $\rho_\lambda$. Thus, we are left with the state
\begin{equation}
    \frac{1}{|G|}\sum_{g\in G}\sum_{\lambda\in\widehat{G}}\sum_{i,j=1}^{d_\lambda}\sqrt{d_\lambda}[\rho_\lambda(g)]^*_{i,j}W(g)E\ket{\psi}_L\ket{\lambda,i,j}_A,
\end{equation}
and now measuring the irrep label yields $\lambda$ with probability $p_\lambda$ given by
\begin{equation}
    \left\|\frac{1}{|G|}\sum_{g\in G}\sum_{i,j=1}^{d_\lambda}\sqrt{d_\lambda}[\rho_\lambda(g)]^*_{i,j}W(g)E\ket{\psi}_L\ket{i,j}_A\right\|^2, 
\end{equation}
or equivalently,
\begin{equation}
\frac{d_\lambda}{|G|^2}\sum_{g,h\in G}\sum_{i,j=1}^{d_\lambda}[\rho_\lambda(g)]^*_{i,j}[\rho_\lambda(h)]_{i,j}\bra{\psi}E^\dagger W(h^{-1}g)E\ket{\psi}_L.
\end{equation}
From the identity
\begin{align}
    \sum_{i,j=1}^{d_\lambda}[\rho_\lambda(g)]_{i,j}^*[\rho_\lambda(h)]_{i,j}&=\tr[(\rho_\lambda(g))^\dagger\rho_\lambda(h)]\\
    &=\chi_\lambda(g^{-1}h),
\end{align}
the probability of obtaining $\lambda$ becomes
\begin{equation}
    p_\lambda=\frac{d_\lambda}{|G|^2}\sum_{g,h\in G}\chi_\lambda(g^{-1}h)\bra{\psi}E^\dagger W(h^{-1}g)E\ket{\psi}_L.
\end{equation}
Now substituting $k=h^{-1}g$, we have
\begin{align}
    p_\lambda&=\frac{d_\lambda}{|G|}\sum_{k\in G}\chi_\lambda(k^{-1})\bra{\psi}E^\dagger W(k)E\ket{\psi}_L\\
    &=\bra{\psi}E^\dagger\Pi_\lambda E\ket{\psi}_L.
\end{align}

Thus, for a state that is in the code space (the $G$-symmetric subspace), the measurement produces the label corresponding to the trivial representation with certainty. Likewise, if the error rotates the logical state entirely into the isotypic component with label $\lambda$, then the measurement produces $\lambda$ with certainty. In either case, the logical register is projected onto the corresponding isotypic component which acts trivially on the logical state. Thus, our measurement acts as a syndrome extraction procedure with syndrome $\lambda$. Moreover, when the error rotates the logical register into a \emph{combination} of isotypic subspaces, the above procedure forces the state to collapse into exactly one isotypic component, so that we observe the discretization of errors characteristic of the stabilizer codes.

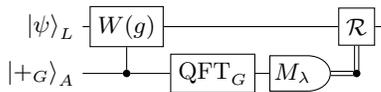
\begin{figure}
    \begin{tikzpicture}
        \begin{yquant}
    
            qubit {$\ket{\psi}_L$} L;
            qubit {$\ket{+_G}_A$} A;
            box {$W(g)$} L | A;
            box {$\qft_G$} A;
            hspace {0.1cm} A;
            dmeter {$M_\lambda$} A;
            box {$\mathcal R$} L | A;
            discard A;
    
        \end{yquant}
    \end{tikzpicture}
    \caption{Isotypic Syndrome Extraction Procedure. The state $\ket{+_G}$ is the uniform superposition over $G$, $M_\lambda$ denotes the measurement with respect to the isotypic label basis, and $\mathcal{R}$ denotes the recovery operation. }
    \label{fig:isotypic-syndrome}
\end{figure}

\subsection{Logical Operations and Errors}\label{sec:logical-ops}

Just as with the stabilizer codes, the logical operations and errors in these generalized codes are classified by taking quotients with respect to various subgroups. Although, since any nonabelian subgroup of the Pauli group contains $-1$ and therefore induces a trivial code space, we must extend the image of the chosen representation beyond the Pauli group to observe interesting phenomenon. Let $\mathcal{Q}$ be a subgroup of $U(\mathcal{H})$ containing the image $W(G)$. The normalizer $N_{\mathcal{Q}}(W(G))$ of $W(G)$ in $\mathcal{Q}$ sits between these two groups so that we have the sequence
\begin{equation}
    W(G)\trianglelefteq N_{\mathcal{Q}}(W(G))\le \mathcal{Q}.
\end{equation}

\begin{proposition}\label{prop:logical}
    The elements of $N_{\mathcal{Q}}(W(G))/W(G)$ are nontrivial logical operations. Operators belonging to the same equivalence class act equivalently on $\sym_G$.
\end{proposition}
\begin{proof}
    Let $N\in N_{\mathcal{Q}}(W(G))$ and let $\ket{\psi}_L\in\sym_G$. Then
    \begin{align}
        \Pi_GN\ket{\psi}_L&=NN^\dagger\Pi_G N\ket{\psi}_L\\
        &=N\frac{1}{|G|}\sum_{g\in G}N^\dagger W(g)N\ket{\psi}_L.
    \end{align}
    Noting that $N^\dagger W(g)N$ is in $W(G)$, it follows that
    \begin{equation}
        \Pi_GN\ket{\psi}_L=N\ket{\psi}_L,
    \end{equation}
    so that $N\ket{\psi}_L$ is in $\sym_G$. Thus, $N$ is a logical operation.

    Suppose now that $N$ and $M$ belong to the same equivalence class in the quotient $N_{\mathcal{Q}}(W(G))/W(G)$. Then $N=MW(g)$ for some $g\in G$. It follows that $N\ket{\psi}_L=MW(g)\ket{\psi}_L=M\ket{\psi}_L$. Thus, operators belonging to the same equivalence class act equivalently on $\sym_G$.
\end{proof}

In the stabilizer code formalism, we choose $\mathcal{Q}$ to be the Pauli group and find that $N_{\mathcal{Q}}(W(G))\trianglelefteq\mathcal{Q}$. The analog is not true in this generalized framework because we no longer have the property that operators which do not commute must anticommute. As a consequence of the normalizer not being normal, the quotient of $\mathcal{Q}$ by the normalizer $N_{\mathcal{Q}}(W(G))$ is not necessarily a group. Moreover, elements of the same left and right cosets of $N_{\mathcal{Q}}(W(G))$ in $\mathcal{Q}$ may not act in the same way on logical states.

Let us refine this observation. Suppose $E$ and $F$ are in the same right coset. Then there is an $N\in N_{\mathcal{Q}}(W(G))$ such that $F=NE$, and it follows that
\begin{equation}
    \bra{\psi}F^\dagger\Pi_\lambda F\ket{\psi}_L=\bra{\psi}E^\dagger N^\dagger\Pi_\lambda NE\ket{\psi}_L.
\end{equation}
In order for $E$ and $F$ to have the same syndrome distribution, we require that $N^\dagger\Pi_\lambda N=\Pi_\lambda$ for all $\lambda$. Exploring this restriction further, we have
\begin{equation}
    N^\dagger\Pi_\lambda N=\frac{d_\lambda}{|G|}\sum_{g\in G}\chi_\lambda(g^{-1})N^\dagger W(g)N\, .
\end{equation}
Recalling that conjugation by $N$ is an automorphism of $W(G)$, we conclude that (when $W$ is faithful) there exists an automorphism $\varphi:G\to G$ such that $W(\varphi(g))=N^\dagger W(g)N$. It therefore follows that
\begin{align}
    N^\dagger\Pi_\lambda N&=\frac{d_\lambda}{|G|}\sum_{g\in G}\chi_\lambda(g^{-1})W(\varphi(g))\\
    &=\frac{d_\lambda}{|G|}\sum_{g\in G}\chi_\lambda((\varphi^{-1}(g))^{-1})W(g),
\end{align}
so that conjugation by an element of the normalizer transforms the isotypic projector into a twisted version of itself. In fact, $\chi_\lambda\circ\varphi^{-1}$ is the character of the representation $\rho_\lambda\circ\varphi^{-1}$, which is irreducible since $\varphi^{-1}$ is an automorphism. Thus, there is a $\mu$ such that $\chi_\lambda\circ\varphi^{-1}=\chi_{\mu}$, and $N^\dagger\Pi_\lambda N=\frac{d_\lambda}{d_\mu}\Pi_{\mu}\, .$ Of course, since $N$ is unitary, we must have $d_\lambda=d_\mu$, and it follows that 
\begin{equation}
    N^\dagger\Pi_\lambda N=\Pi_{\mu}\, .
\end{equation}
As a consequence, an element of the normalizer can rotate a state from one nontrivial isotypic component to another; however, the trivial isotypic component (the code space) is invariant under this action as shown in Proposition~\ref{prop:logical}.

The preceeding discussion shows that there are two classes of logical operations: those that commute with $W(G)$ and therefore leave isotypic subspaces invariant, and those that permute the non-trivial isotypic subspaces. Ideally, no error is applied to our logical qubit, and so this rotation of non-trivial isotypic subspace by elements outside of the commutant of $W(G)$ is inconsequential. However, when such an error does happen, the state is pushed into a non-trivial isotypic subspace and subsequent applications of logical operations further permute the isotypic subspace if such operations are not in the commutant. Thus, in order to apply a recovery operation, we must keep track of both the isotypic syndromes measured and the non-commutant logical operations performed.

Let us determine which errors can occur. To ensure that errors in the same coset produce the same syndrome, we must restrict our attention from the normalizer $N_{\mathcal{Q}}(W(G))$ to the subgroup $\mathcal{N}_0$ of elements which fix the irrep label under conjugation. This motivates our next proposition:

\begin{proposition}
    Let $E,F\in U(\mathcal{H})$ be errors. If $E$ and $F$ belong to the same right coset of $\mathcal{N}_0$ in $U(\mathcal{H})$, then they have the same syndrome distribution.
\end{proposition}
\begin{proof}
    If $E$ and $F$ belong to the same right coset of $\mathcal{N}_0$ in $U(\mathcal{H})$, then $E=NF$ for some $N\in\mathcal{N}_0$. It follows that $\bra{\psi}E^\dagger\Pi_\lambda E\ket{\psi}_L=\bra{\psi}F^\dagger N^\dagger\Pi_\lambda NF\ket{\psi}_L=\bra{\psi}F^\dagger \Pi_\lambda F\ket{\psi}_L$.
\end{proof}

Since our symmetry-based construction operates in the same vein as a stabilizer code, it retains the famous Knill-Laflamme conditions \cite{knill2000}, which we restate here for completeness.

\begin{theorem}
    Let $\mathcal{E}$ be a set of linear operators acting on the physical Hilbert space $\mathcal{H}$. Then all elements of $\mathcal{E}$ are detectable and correctable by a $G$-code if and only if for each $E\in\mathcal{E}^\dagger\mathcal{E}$, there is a constant $c_E$ depending on $E$ such that
    \begin{equation}
        \Pi_G E\Pi_G=c_E\Pi_G.
    \end{equation}
\end{theorem}

Since the logical operators are no longer confined to the Pauli group, we must generalize our notion of distance to capture the new logical operations that can arise. Consider the case of $n$ physical qubits, so that the Hilbert space is $\mathcal{H}=\mathbb{C}^{2^n}$, and label the qubits $\{0,\ldots,n-1\}$. Choose a generating set $\mathcal{G}$ for $Q$, and let $E\in Q$. Then $E=g_1\cdots g_N$ for some $g_i\in\mathcal{G}$. We define the $G$-weight $w_G(E)$ to be the minimum number $N$ needed to generate $E$. The $G$-distance $d_G$ of the code is then the minimum weight of a nontrivial logical operation:
\begin{equation}
    d_G:=\min_{E\in N_Q(W(G))\setminus W(G)}w_G(E).
\end{equation}
We recover the usual Pauli distance for stabilizer codes by choosing $Q$ to be the Pauli group on $n$ qubits with generating set $\mathcal{G}$ consisting of all single qubit Pauli operators. With this notion of distance, we have the following generalized Knill-Laflamme bound.
\begin{theorem}[Generalized Knill-Laflamme bound]
\label{thm:generalized-KL}
Let $Q \leq U(\mathcal{H})$ and $\mathcal{G}$, $w_G(\cdot)$, $d_G$ be as above.
Assume that the weight $w_G$ satisfies the triangle inequality
\begin{equation}
    w_G(AB) \leq w_G(A) + w_G(B)
\end{equation}
for all $A,B \in Q$. Let
\begin{equation}
    \mathcal{E}_t := \{ E \in N_Q(W(G)) : w_G(E) \leq t \}
\end{equation}
be a set of symmetry-preserving errors of $G$-weight at most $t$, and suppose that
\begin{equation}
    2t < d_G.
\end{equation}
Then the Knill--Laflamme condition holds for $\mathcal{E}_t$ and the codespace $\sym_G$:
for every pair $E,F \in \mathcal{E}_t$ there exists a scalar $c_{E,F} \in \mathbb{C}$
such that
\begin{equation}
    \Pi_{G} E^\dagger F \,\Pi_{G} \;=\; c_{E,F}\, \Pi_{G}.
\end{equation}
In particular, any subset of errors of $G$-weight at most
$\lfloor (d_G - 1)/2 \rfloor$ is correctable.
\end{theorem}

\begin{proof}
Let $E,F \in \mathcal{E}_t$ and set $A := E^\dagger F \in Q$.  
By definition of $\mathcal{E}_t$ we have $E,F \in N_Q(W(G))$ and
$w_G(E),w_G(F) \leq t$, so that
\begin{align}
    w_G(A) &= w_G(E^\dagger F)\\
    &\leq w_G(E^\dagger) + w_G(F)\\
    &= w_G(E) + w_G(F)\\
    &\leq 2t< d_G,
\end{align}
where we used the triangle inequality and the fact that $w_G(E^\dagger)=w_G(E)$.

Since $E,F \in N_Q(W(G))$ and $N_Q(W(G))$ is a group, we also have
$A = E^\dagger F \in N_Q(W(G))$.  There are now two possibilities:

\smallskip
\noindent
\textbf{(i) $A \in W(G)$.}
In this case $A$ acts as the trivial representation on $\sym_G$.  Therefore
\begin{equation}
    \Pi_{G} A \Pi_{G} = \Pi_{G},
\end{equation}
and we may take $c_{E,F} = 1$ in the Knill-Laflamme condition.

\smallskip
\noindent
\textbf{(ii) $A \notin W(G)$.}
Since $A \in N_Q(W(G))\setminus W(G)$ and $d_G$ is by definition the minimum
$G$-weight over such elements, we must have
\begin{equation}
    w_G(A) \;\geq\; d_G,
\end{equation}
contradicting $w_G(A) < d_G$.

We conclude that only case (i) is possible: for every $E,F\in\mathcal{E}$,
$E^\dagger F \in W(G)$ and acts as the identity on $\sym_G$.
Thus $\Pi_{G} E^\dagger F \Pi_{G} = \Pi_{G}$,
which is the desired Knill-Laflamme condition with $c_{E,F}=1$. Thus for any pair of errors $E,F$ of $G$-weight at most $t$, either $E$ and $F$ act identically on $\sym_G$ (differing only by an element of $W(G)$), or the subspaces $E(\sym_G)$ and $F(\sym_G)$ are orthogonal.  This is precisely the Knill--Laflamme criterion for error correction: the images of the codespace under distinct errors are mutually orthogonal unless the errors have the same logical effect. Therefore a recovery map may be defined by projecting onto each orthogonal error subspace and inverting the corresponding operator. Consequently, every error of $G$-weight at most $t$ is correctable, and in particular all errors of $G$-weight at most $\lfloor (d_G-1)/2\rfloor$ are correctable.

\end{proof}

\subsection{Two Logical Qubit Gates}\label{sec:two-logical}
When designing fault-tolerant architectures, it is convenient to design codes for a single logical qubit such that errors are local in the code space. That is, an error affecting one logical qubit does not necessarily correlate to an error affecting another logical qubit. Of course, the drawback is that universal quantum computation cannot be done with a single logical qubit. Thus, a method has to be created for performing logical operations between the distinct logical qubits, each of which are protected by their own code. Perhaps the most effective way to do this is with lattice surgery in the surface code setting \cite{horsman2012,fowler2019,litinski2019}. We take a similar approach here, joining the logical qubits with an ancillary logical qubit and then performing measurements that couple the systems. The desired logical operation is then recovered after performing a correction operation based on the measurement outcomes.

Take, for example, the $\cnot$ gate. %from one logical qubit to another. 
Denote the control register by $C$, the target register by $T$, and the state of the joint system by $\ket{\psi}_{CT}$. To implement the $\cnot_{CT}$ gate, we first prepare an ancillary logical qubit in the $\ket{+}_A$ state. Then we measure the observable $Z_CZ_A$, followed by the observable $X_AX_T$, and finally the observable $Z_A$. These measurements produce outcomes $a\, ,b,$ and $c$, respectively, and the state (up to normalization) after all three measurements is 
\begin{equation}
    (1+cZ_A)(1+bX_AX_T)(1+aZ_CZ_A)\ket{\psi}_{CT}\ket{+}_A.
\end{equation}
A copious amount of algebra shows that this state is %(up to normalization) 
equivalent to
\begin{equation}
\begin{aligned}
    &(1+c)(1+bX_T+aZ_C-abX_TZ_c)\ket{\psi}_{CT}\ket{0}_A\\
    &+(1-c)(1+bX_T-aZ_C+abX_TZ_C)\ket{\psi}_{CT}\ket{0}_A.
\end{aligned}
\end{equation}
Noting that $\cnot_{CT}=\frac{1}{2}(1+Z_C+X_T-Z_CX_T)$, we see that the different possible measurement outcomes produce the $\cnot$ gate up to a possible correction. Table~\ref{tab:correction-cnot}, lists the corrections needed to obtain the desired logical $\cnot$ gate.

\begin{table}[]
    \centering
    \begin{tabular}{|c|c|c|c|}
    \hline
    a & b & c & Correction \\
    \hline
    $+1$ & $+1$ & $+1$ & - \\
    $+1$ & $+1$ & $-1$ & $X_T$ \\
    $+1$ & $-1$ & $+1$ & $Z_C$ \\
    $+1$ & $-1$ & $-1$ & $X_TZ_C$ \\
    $-1$ & $+1$ & $+1$ & $X_T$ \\
    $-1$ & $+1$ & $-1$ & - \\
    $-1$ & $-1$ & $+1$ & $X_TZ_C$ \\
    $-1$ & $-1$ & $-1$ & $Z_C$ \\
    \hline
\end{tabular}
    \caption{Correction table for logical $\cnot$ gate.}
    \label{tab:correction-cnot}
\end{table}

This procedure is independent of the choice of group and representation. Indeed, on a logical-qubit level, there is no reference to either of these objects. They become relevant only when performing measurements on a physical level, in which case the logical $X$ and $Z$ operations need to be computed based on the choice of group representation. Section~\ref{sec:alternative-sym} gives an example. Labeling our basis elements $\ket{0}_L=\ket{00}$ and $\ket{1}_L=\frac{1}{\sqrt{3}}(\ket{01}+\ket{10}+\ket{11})$, it can be checked that suitable logical Pauli operations are
\begin{align}
    X_L&=\frac{1}{2\sqrt{3}}(X_1+X_2+X_1X_2+X_1Z_2+Z_1X_2-Y_1Y_2)\\
    Z_L&=\frac12(Z_1+Z_2+Z_1Z_2-1).
\end{align}

\subsection{Rate Estimates}
The number of physical qubits in the code defined by the finite-dimensional representation $W:G\to U(\mathcal{H})$ is $n=\log_2(\dim(\mathcal{H}))$. Since our code space is the $G$-symmetric subspace, the maximum number of logical qubits is $\lfloor\log_2(\dim(\sym_G))\rfloor$. Thus, the rate is bounded above by
\begin{equation}\label{eq:rate-bound}
    R\le\frac{\lfloor\log_2(\dim(\sym_G))\rfloor}{\log_2(\dim(\mathcal{H}))}.
\end{equation}
For the tensor permutation representation of the symmetric group $S_n$, the symmetric subspace is spanned by the symmetric Dicke states of weights up to $n$. There are $n+1$ such states, so the rate of this code is $R=\lfloor\log_2(n+1)\rfloor/n$. 

In general, the dimension of the symmetric subspace given a permutation representation $W$ of a finite group $G$ is $Z_G(d,\ldots,d)$, where $d$ is the number of qubits acted on by $G$ and $Z_G$ is the cycle index polynomial of $G$, which has appeared as a combinatorial tool in a number of recent quantum information theoretic works \cite{bradshaw2023,bradshaw2023thesis,bradshaw2023zeta,bradshaw2025moments}. This result is a direct consequence of Polya's enumeration theorem \cite{polya1937,bradshaw2024}. 

With complete generality, the dimension of the symmetric subspace is given by the rank of the corresponding projector or, equivalently, its trace
\begin{align}
    \dim(\sym_G)&=\tr[\Pi_G]\\
    &=\frac{1}{|G|}\sum_{g\in G}\tr[W(g)]\\
    &=\frac{1}{|G|}\sum_{g\in G}\chi_W(g).
\end{align}
In other words, the dimension of the symmetric subspace is the average value of the character of the chosen representation over the represented group. Letting $W=\bigoplus_{\lambda}\rho_\lambda^{\oplus m_\lambda}$, we have $\chi_W=\sum_{\lambda}m_\lambda\chi_\lambda$ so that
\begin{align}
    \dim(\sym_G)&=\frac{1}{|G|}\sum_{g\in G}\sum_{\lambda}m_\lambda\chi_\lambda(g)\\
    &=\frac{1}{|G|}\sum_{\lambda}m_\lambda\sum_{g\in G}\chi_\lambda(g).
\end{align}
Imposing the character orthogonality relations, we see that the only term in the sum over $\lambda$ that doesn't vanish is the one corresponding to the trivial representation, and it follows that
\begin{align}
    \dim(\sym_G)&=m_{\textnormal{triv}},
\end{align}
where $m_{\textnormal{triv}}$ is the multiplicity of the trivial representation in $W$. Since the symmetric subspace is the isotypic component corresponding to the trivial representation, this is indeed what we expect. In terms of the multiplicities and dimensions of each irrep, the bound \eqref{eq:rate-bound} therefore becomes
\begin{equation}
    R\le\frac{\lfloor\log_2(m_{\textnormal{triv}})\rfloor}{\log_2(\sum_\lambda m_\lambda d_\lambda)}.
\end{equation}

\section{Examples}\label{sec:examples}

\subsection{Stabilizer Codes}\label{sec:stabilizer-codes}
We will now show that all stabilizer codes are $G=\mathbb Z_2^{\oplus n}$ codes for some $n$. Indeed, the generators $S_1,\ldots,S_n$ of the stabilizer subgroup $\mathcal{S}$ are chosen to be elements of the Pauli group which square to the identity and commute with one another. Thus, there is an isomorphism $W:\mathbb Z_2^{\oplus n}\to\mathcal{S}$ given by $W(k_1,\ldots,k_n)=\prod_{i=1}^nS_i^{k_i}$. The code space is taken to be the space of all states fixed by the action of the stabilizer subgroup, which is the symmetric subspace $\sym_{\mathbb Z_2^{\oplus n}}$. Here, we have taken the action of $\mathbb Z_2^{\oplus n}$ to be induced by the isomorphism $W$.

We claim that the usual syndrome extraction procedure given by performing a Hadamard test for each generator of $\mathcal S$ is identical to our procedure for identifying the isotypic subspace to which a state belongs. Since $\mathbb Z_2^{\oplus n}$ is abelian, its irreducible representations are 1-dimensional and identical to the corresponding character. Specifically, they are given by $\chi_j(k)=(-1)^{\sum_{i}j_ik_i}$ for all $j,k\in\mathbb Z_2^{\oplus n}$. There are $2^n$ such choices of $j$, and the number of distinct irreducible representations of a finite group (over $\mathbb C$) is equivalent to the number of conjugacy classes of the group ($2^n$ for an abelian group of order $2^n$). Thus, we have written all possible irreps: $\rho_j(k)=\chi_j(k)=(-1)^{\sum_{i}j_ik_i}=(-1)^{j\cdot k}$. Our goal is therefore to show that the usual syndrome extraction procedure given by a sequence of Hadamard tests makes the projective measurement $\bra{\psi}\Pi_j \ket{\psi}_L$ on the logical state, possibly after an error. To this end, we begin with the state $\ket{\psi}_L\ket{0}^{\otimes n}_{A_1}$, and apply a Hadamard gate to the ancillary register, creating the state
\begin{equation}
    \ket{\psi}_L\frac{1}{2^{n/2}}(\ket{0}+\ket{1})^{\otimes n}=\frac{1}{2^{n/2}}\sum_{k\in\mathbb{Z}_2^{\oplus n}}\ket{\psi}_L\ket{k}_A,
\end{equation}
where $\ket{k}_A$ is the computational basis state $\ket{k_1\cdots k_n}$. Now, apply a series of $S_i$ operations each controlled off of the $i$-th ancillary qubit to obtain
\begin{equation}
    \frac{1}{2^{n/2}}\sum_{k\in\mathbb{Z}_2^{\oplus n}}S_1^{k_1}\cdots S_n^{k_n}\ket{\psi}_L\ket{k}_A.
\end{equation}
Next, apply a Hadamard gate to each qubit in the ancillary register again to achieve the final state 
\begin{equation}
    \frac{1}{2^n}\sum_{k\in\mathbb{Z}_2^{\oplus n}}S_1^{k_1}\cdots S_n^{k_n}\ket{\psi}_L\prod_{i=1}^n(\ket{0}+(-1)^{k_i}\ket{1}),
\end{equation}
which is equivalent to
\begin{equation}
    \frac{1}{2^n}\sum_{k\in\mathbb{Z}_2^{\oplus n}}W(k)\ket{\psi}_L\sum_{j\in\mathbb{Z}_2^{\oplus n}}(-1)^{j\cdot k}\ket{j}_A.
\end{equation}
The ancillary register now labels the irreps. Making a measurement of this register produces some outcome $j$, in which case, the resulting state of the logical register is
\begin{equation}
    \frac{1}{2^n}\sum_{k\in\mathbb{Z}_2^{\oplus n}}(-1)^{j\cdot k}W(k)\ket{\psi}_L=\Pi_j\ket{\psi}_L.
\end{equation}
Thus, the probability that outcome $j$ occurs is $\bra{\psi}\Pi_j\ket{\psi}_L$, showing that the syndrome extraction for the stabilizer code is identical to the isotypic syndrome extraction we have outlined above. Notably, the QFT over the group in this case was given by a tensor product of Hadamard gates.

\subsection{Qudit Stabilizer Codes}

A slight generalization of the stabilizer code can be produced by allowing generators which still commute but, instead of squaring to the identity, now resolve to the identity at the $d$-th power. Let $\{S_1,\ldots,S_n\}$ be a set of independent generators of an abelian subgroup of $U(\mathcal{H})$, none of which are $e^{2\pi i/d}I$, which we again call $\mathcal S$. Assume $\mathcal{H}$ is a qudit state space with $\dim(\mathcal{H})=d^N$ for some integer $N$. We then represent $\mathbb Z_d^{\oplus n}$ with $W:\mathbb Z_d^{\oplus n}\to\mathcal S$ defined by $W(k)=\prod_{i=1}^nS_i^{k_i}$, where $k=(k_1,\ldots,k_n)\in\mathbb Z_d^{\oplus n}$. The code space is taken to be the $\mathbb Z_d^{\oplus n}$-symmetric subspace $\sym_{\mathbb Z_d^{\oplus n}}$ with respect to the action induced by $W$.

We now perform isotypic syndrome extraction as outlined above. We first prepare the state \begin{equation}
    \frac{1}{d^{n/2}}\sum_{k\in \mathbb Z_d^{\oplus n}}\ket{\psi}_L\ket{k} \, , 
\end{equation}
where $\ket{k}$ is the computational basis state that has the base $d$ representation $k_n\cdots k_1$. That is, $\ket{k}$ is the $(k_1+dk_2+\cdots+d^{n-1}k_n)$-th computational basis state. We now apply a controlled $W(k)=\prod_{i=1}^nS_i^{k_i}$ operation targeting the logical register and controlled off of the ancillary register, triggering when the ancillary state is $\ket{k}$. The resulting state is
\begin{equation}
    \frac{1}{d^{n/2}}\sum_{k\in\mathbb Z_d^{\oplus n}}S_1^{k_1}\cdots S_n^{k_n}\ket{\psi}_L\ket{k},
\end{equation}
and now it remains to apply the QFT over $\mathbb Z_d^{\oplus n}$ and make a measurement with respect to the irrep basis. To accomplish this, we need to understand the irreps of $\mathbb{Z}_d^{\oplus n}$, which we note are 1-dimensional since this group is abelian. We define them by $\chi_j(k)=\omega^{j\cdot k}$, where $\omega=e^{2\pi i/d}$ is the principal $d$-th root of unity, and $j,k\in\mathbb{Z}_d^{\oplus n}$. Since there are $d^n$ such choices of $j$ and there are similarly $d^n$ conjugacy classes of $\mathbb{Z}_d^{\oplus n}$, we have found all the irreps for this group. Thus, the QFT over $\mathbb{Z}_d^{\oplus n}$ is the operator defined by
\begin{equation}
    \qft_{\mathbb{Z}_d^{\oplus n}}\ket{k} = \frac{1}{d^{n/2}}\sum_{j\in\mathbb{Z}_d^{\oplus n}}\omega^{-j\cdot k}\ket{j},
\end{equation}
which is a product of ordinary QFTs on each qudit. Applying this operation to the ancillary register produces the state
\begin{equation}
    \frac{1}{d^n}\sum_{k,j\in\mathbb{Z}_d^{\oplus n}}\omega^{-j\cdot k}S_1^{k_1}\cdots S_n^{k_n}\ket{\psi}_L\ket{j}_A.
\end{equation}
Now, measuring the ancillary register produces some label $j$, and the state of the logical register immediately after this measurement is
\begin{equation}
    \frac{1}{d^n}\sum_{k\in\mathbb{Z}_d^{\oplus n}}\chi_j(-k)W(k)\ket{\psi}_L=\Pi_j\ket{\psi}_L,
\end{equation}
where we have used the facts that $\chi_j(k)=\omega^{-j\cdot k}$ and $W(k)=S_1^{k_1}\cdots S_n^{k_n}$.

If no error occurs on the logical state before syndrome extraction, then $W(k)\ket{\psi}_L=\ket{\psi}_L$ for all $k$, and we are left with the state
\begin{equation}
     \frac{1}{d^n}\sum_{k\in\mathbb{Z}_d^{\oplus n}}\chi_j(-k)\ket{\psi}_L=\begin{cases}
         \ket{\psi}_L, & j=(0,\ldots,0)\\
         0, & j\ne (0,\ldots, 0)
     \end{cases}
\end{equation}
by the character orthogonality relations. Thus, we produce the syndrome $j=(0,\ldots,0)$ with certainty whenever an error does not occur.

More generally, if an error $E$ occurs so that the logical state is of the form $E\ket{\psi}_L$ for some $\ket{\psi}_L\in\sym_{\mathbb{Z}_d^{\oplus n}}$, then we obtain outcome $j$ with probability
\begin{equation}
    p_j=\bra{\psi}E^\dagger\Pi_jE\ket{\psi}_L,
\end{equation}
where \begin{equation}
   \Pi_j=\frac{1}{d^n}\sum_{k\in\mathbb{Z}_d^{\oplus n}}\chi_j(-k)W(k)=\frac{1}{d^n}\sum_{k\in\mathbb{Z}_d^{\oplus n}}\omega^{-j\cdot k}S_1^{k_1}\cdots S_n^{k_n}\, . 
\end{equation}
Using the latter form of the projector, we can translate our construction into the usual language of a stabilizer code. Let $\ell\in\mathbb{Z}_d^{\oplus n}$ and assume $S_iE=\omega^{\ell_i}ES_i$. This is the qudit analog of the restriction that Pauli operations must either commute or anticommute. Then we have
\begin{equation}
    \Pi_jE\ket{\psi}_L=\frac{1}{d^n}\sum_{k\in\mathbb{Z}_d^{\oplus n}}\omega^{-j\cdot k}\omega^{\ell\cdot k}\ket{\psi}_L,
\end{equation}
where we have used the fact that $S_i\ket{\psi}_L=\ket{\psi}_L$ for all $i$. Now rewriting in terms of characters and invoking the character orthogonality relations reveals that
\begin{align}\label{eq:qudit-projection}
    \Pi_jE\ket{\psi}_L&=\frac{1}{d^n}\sum_{k\in\mathbb{Z}_d^{\oplus n}}\chi_j(-k)\chi_\ell(k)\ket{\psi}_L\\
    &=\begin{cases}
        \ket{\psi}_L, & j=\ell\\
        0,&j\ne\ell.
    \end{cases}
\end{align}
Thus, for each choice of $\ell$, the logical state is rotated into the isotypic subspace with label $\ell$, and our syndrome extraction procedure predicts this label with certainty.

One reason stabilizer codes are useful is because they correct a subset of the operations in the Pauli group, which form a basis for all unitary operators acting on the physical Hilbert space. These codes work because the Pauli operators are unitary and square to the identity, allowing eigenvalues of $\pm1$ only. Thus, all elements of the Pauli group either commute or anticommute. Moreover, if the error is a superposition of elements of the Pauli group, we observe the discretization of errors, wherein the projective measurement made during syndrome extraction forces the logical state to effectively choose which element of the Pauli group it was affected by. A similar phenomenon happens in the $G=\mathbb{Z}_d^{\oplus n}$ code. Indeed, we consider a generalized Pauli group $\mathcal{P}_{d,n}$ acting on $n$-fold tensor products of $\mathbb{C}^d$. This group is generated by tensor products of the generalized Pauli matrices constructed by Sylvester in 1882 \cite{sylvester2008collected}. These $d^2$ independent unitary matrices are equivalently defined by
\begin{equation}
    \sigma_{k,j} = \sum_{m=0}^{d-1}\omega^{jm}\ket{m+k}\!\!\bra{m},
\end{equation}
where $\omega=e^{2\pi i/d}$ is the principal $d$-th root of unity. In addition to carrying a natural group structure as a matrix representation of the Weyl-Heisenberg group on a $d$-dimensional Hilbert space, these matrices carry the structure of a generalized Clifford algebra \cite{weyl1927,sylvester2008collected,cartan1993groupes,schwinger1960,santhanam1976,ablamowicz1996,kwasniewski1999} (not to be confused with the notion of a generalized Clifford algebra induced by higher order forms \cite{albertin2023,chapman2015,childs1978,garling2011,pappacena2000,revoy2014}). Indeed, they satisfy the relations $\sigma_{k,j}\sigma_{k',j'}=\omega^{k'j-kj'}\sigma_{k',j'}\sigma_{k,j}$, which follow from the fact that these matrices are generated by the clock and shift matrices. We take these latter matrices as our qudit $X$ and $Z$ gates \cite{sarkar2024}:
\begin{equation}
    X=\sigma_{1,0}=\sum_{m=0}^{d-1}\ket{m+1}\!\!\bra{m},
\end{equation}
\begin{equation}
    Z=\sigma_{0,1}=\sum_{m=0}^{d-1}\omega^{m}\ket{m}\!\!\bra{m}.
\end{equation}
Thus, the generators satisfy the identity $ZX=\omega XZ$, a relation which defines a generalized Clifford algebra analogous to the Clifford algebra induced by the anticommutation relation between the Pauli $X$ and $Z$ gates in a 2-qubit system.

Returning to our code, let us examine what happens if an error $E=\sigma_{k,j}$ occurs. If our stabilizer generators $S_i$ are selected from the generalized Pauli group $\mathcal{P}_{d,n}$, then it follows that there is an $\ell=(\ell_1,\ldots,\ell_n)\in\mathbb{Z}_d^{\oplus n}$ such that $S_iE=\omega^{\ell_i}ES_i$ for each $i$. Thus, \eqref{eq:qudit-projection} guarantees that the outcome of our syndrome measurement will be the label $\ell$ corresponding to the isotypic subspace with projector $\Pi_\ell$. We have therefore constructed a generalization of the stabilizer code to qudit systems by simply choosing a natural representation of $G=\mathbb{Z}_d^{\oplus n}$.

In qubit stabilizer codes, it is well-known that the addition of an independent stabilizer generator has the effect of halving the dimension of the code space \cite{bradshaw2025}. As we now show, there is a similar relationship for qudit stabilizer codes \cite{gheorghiu2014}.
\begin{proposition}\label{prop:dim-stab}
    Consider an $N$-qudit stabilizer code with stabilizer $\mathcal{S}=\mathbb{Z}_d^{\oplus n}$ consisting of $n$ independent generators. Then $\dim(\sym_{\mathbb{Z}_d^{\oplus n}})  |\mathbb{Z}_d^{\oplus n}| = d^N$, where d is the dimension of the Hilbert space associated to a single qudit. That is, the dimension of the code space is $d^{N-n}$.
\end{proposition}
\begin{proof}
    The projection onto $\sym_{\mathbb{Z}_d^{\oplus n}}$ is given by
    \begin{align}
        \Pi&=\frac{1}{|\mathbb{Z}_d^{\oplus n}|}\sum_{k\in\mathbb{Z}_d^{\oplus n}}W(k)\\
        &=\frac{1}{d^n}\sum_{k\in\mathbb{Z}_d^{\oplus n}}S_1^{k_1}\cdots S_n^{k_n}
    \end{align} Since $\tr[\Pi]=\rank(\Pi)=\dim(\sym_{\mathbb{Z}_d^{\oplus n}})$, and since each $S_i$ is taken to be an element of $\mathcal{P}_{d,n}$, all of which have trace zero except multiples of the identity, we therefore have $\dim(\sym_{\mathbb{Z}_d^{\oplus n}})=\frac{d^N}{d^n}=d^{N-n}$ since the $S_i$ are independent.
\end{proof}
It follows from Proposition~\ref{prop:dim-stab} that a stabilizer code on a $d^N$-dimensional physical Hilbert space with $N-1$ independent stabilizers will represent a single logical qudit. Thus, the rate of such a code is $R=\frac{1}{N}$. The distance, however, will vary greatly depending on the chosen generators of the stabilizer subgroup.

\subsection{Mixed Dimension Stabilizer Codes}

Suppose we have a mix of qudits of many different dimensions, say qudit $i$ has dimension $d_i$. The resulting Hilbert space is $\mathcal{H}=\mathbb{C}^{d_1}\otimes\cdots\otimes\mathbb{C}^{d_n}$. We represent $G=\mathbb{Z}_{d_1}\oplus\cdots\oplus\mathbb{Z}_{d_n}$ with some choice of representation $W:G\to\mathcal{P}_{d_1}\otimes\cdots\otimes\mathcal{P}_{d_n}\subset U(\mathcal{H})$, where $\mathcal{P}_{d}$ denotes the qudit Pauli group on a single qudit. In particular, we take $W(k_1,\ldots,k_n)=S_1^{k_1}\otimes\cdots\otimes S_n^{k_n}$, where $k_i\in\mathbb{Z}_{d_i}$. We again take the code space to be the $G$-symmetric subspace $\sym_G$ with respect to the action induced by $W$. The actual form of this subspace depends on the choice of $S_i$.

We now perform isotypic syndrome extraction. First, prepare the state $\frac{1}{\sqrt{d_1\cdots d_n}}\sum_{(k_1,\ldots,k_n)\in G}\ket{\psi}_L\ket{k_1\cdots k_n}$, where $\ket{\psi}_L$ is the state of the logical register before syndrome extraction. Now, apply a controlled $W(k_1,\ldots,k_n)=S_1^{k_1}\otimes\cdots\otimes S_n^{k_n}$ operation targeting the logical register and controlled off of the ancillary register and triggered when the ancillary state is $\ket{k_1\cdots k_n}$. This results in the state
\begin{equation}
    \frac{1}{\sqrt{d_1\cdots d_n}}\sum_{(k_1,\ldots, k_n)\in G}S_1^{k_1}\otimes\cdots\otimes S_n^{k_n}\ket{\psi}_L\ket{k_1\cdots k_n}.
\end{equation}
The next step is to apply the QFT over $G$. This group is again abelian and the irreps are given by $\chi_{j_1,\ldots,j_n}(k_1,\ldots,k_n)=\omega_{d_1}^{j_1k_1}\cdots\omega_{d_n}^{j_nk_n}$, where $\omega_d=e^{2\pi i/d}$ denotes the $d$-th principal root of unity. The QFT is therefore defined by
\begin{align}
    &\qft_{\mathbb{Z}_{d_1}\oplus\cdots\oplus\mathbb{Z}_{d_n}}\ket{k_1\cdots k_n}\notag\\
    &=\frac{1}{\sqrt{d_1\cdots d_n}}\sum_{(j_1,\ldots,j_n)\in G}\omega_{d_1}^{-j_1k_1}\cdots\omega_{d_n}^{-j_nk_n}\ket{j_1\cdots j_n}\\
    &=\prod_{i=1}^n\frac{1}{\sqrt{d_i}}\sum_{j_i\in\mathbb{Z}_{d_i}}\omega_{d_i}^{-j_ik_i}\ket{j_i}\\
    &=\qft_{\mathbb{Z}_{d_1}}\ket{k_1}\otimes\cdots\otimes\qft_{\mathbb{Z}_{d_n}}\ket{k_n}.
\end{align}
Applying this operation to the ancillary register produces the state
\begin{equation}\label{eq:mixed-qudit-post-state}
    \frac{1}{d_1\cdots d_n}\sum_{(j_1,\ldots,j_n)\in G}\sum_{(k_1,\ldots,k_n)\in G}\left(\bigotimes_{i=1}^n\omega_{d_i}^{-j_ik_i}\ket{j_i}_AS_i^{k_i}\right)\ket{\psi}_L,
\end{equation}
where we have rearranged the registers for notational convenience. Finally, measuring the ancillary register produces a label $j_1\cdots j_n$, and the state of the logical register immediately after this measurement is $\Pi_{j_1\cdots j_n}\ket{\psi}_L$.

Denote by $X_i$ and $Z_i$ the generalized Pauli $X$ and $Z$ gates on qudit $i$ of the ancillary register. Thus, $X_i$ and $Z_i$ are $d_i\times d_i$ unitary matrices for each $i=1,\ldots,n$. When no errors occurs, \eqref{eq:mixed-qudit-post-state} becomes
\begin{equation}
    \frac{1}{d_1\cdots d_n}\sum_{(j_1,\ldots,j_n)\in G}\sum_{(k_1,\ldots,k_n)\in G}\left(\bigotimes_{i=1}^n\omega_{d_i}^{-j_ik_i}\ket{j_i}_A\right)\ket{\psi}_L,
\end{equation}
and a measurement of the ancillary register produces the state
\begin{equation}
    \frac{1}{d_1\cdots d_n}\sum_{(k_1,\ldots,k_n)\in G}\left(\bigotimes_{i=1}^n\omega_{d_i}^{-j_ik_i}\ket{j_i}_A\right)\ket{\psi}_L,
\end{equation}
for some label $j_1\cdots j_n$. By the character orthogonality relations, this is the zero state whenever there is an $i$ such that $j_i\ne0$. Otherwise, the logical state is recovered with no error. Now suppose an $X_m$ error has been applied to the logical state before syndrome extraction. Then $S_mX_m=\omega_{d_m}^{\ell_m}X_mS_m$, and it follows that the state immediately before a measurement of the ancillary register is
\begin{equation}
    \frac{1}{d_1\cdots d_n}\sum_{\substack{(j_1,\ldots,j_n)\in G\\(k_1,\ldots,k_n)\in G}}\omega_{d_m}^{\ell_mk_m}\left(\bigotimes_{i=1}^n\omega_{d_i}^{-j_ik_i}\ket{j_i}_A\right)X_i\ket{\psi}_L.
\end{equation}
Thus, from the character orthogonality relations, a measurement of the ancillary register produces the syndrome $(0,\ldots,0,\ell_m,0,\ldots,0)$ with certainty, and the state immediately after this measurement is $\ket{0\cdots0\ \ell_m\ 0\cdots0}_AX_m\ket{\psi}_L$. We may therefore correct the Pauli-$X$ error on the $m$-th qudit by applying another Pauli-$X$ to the same qudit. This process can be repeated with any other tensor product of Pauli errors, and there are $d_1\cdots d_n$ possible syndromes, including the all zero syndrome indicative of no error occurring. The code therefore distinguishes $d_1\cdots d_n-1$ nontrivial syndromes.

\subsection{Symmetric Group Permutation Representation}\label{sec:ex_S3}
Let us demonstrate our general construction with the simplest nonabelian group, $G=S_3$, acting on three qubits by permuting tensor factors. The physical Hilbert space is $\mathcal H=(\mathbb C^2)^{\otimes 3}$ with computational basis $\{\ket{ijk}\}_{i,j,k=0}^1$.

The irreducible representations of $S_3$ (over $\mathbb{C}$) are the 1-dimensional trivial representation defined by $\rho_{\mathrm{triv}}(\sigma)=1$, the 1-dimensional sign representation $\rho_{\mathrm{sgn}}(\sigma)=\mathrm{sgn}(\sigma)$, and a 2-dimensional representation called the standard representation. This latter representation is defined similarly to the tensor permutation representation, except we permute the basis of the Hilbert space rather than the tensor components. This basis permutation representation is not itself irreducible because it contains a factor of the trivial representation, but by restricting this representation to the quotient of the total Hilbert space $(\mathbb{C}^2)^{\otimes 3}$ by the invariant subspace $\langle \sum_{i,j,k=0}^1\ket{ijk}\rangle$, we obtain the standard representation. For the 1-dimensional representations, the characters are precisely the representation itself. The character of the standard representation maps the identity element to $2$, transpositions to $0$, and 3-cycles to $-1$. Using Lemma~\ref{lemma:lambda-projection}, the isotypic projection operators are
\begin{equation}
    \Pi_{\mathrm{triv}}=\frac{1}{6}\sum_{\sigma\in S_3}W(\sigma)
\end{equation}
for the trivial representation,
\begin{equation}
    \Pi_{\mathrm{sgn}} = \frac{1}{6}\sum_{\sigma\in S_3}\mathrm{sgn}(\sigma^{-1})W(\sigma),
\end{equation}
for the sign representation, and
\begin{equation}
    \Pi_{\mathrm{std}} = \frac{1}{3}\sum_{\sigma\in S_3}\chi_{\mathrm{std}}(\sigma^{-1})W(\sigma)
\end{equation}
for the standard representation.

For the tensor permutation action of $S_3$ on three qubits, we find from the character orthogonality relations that the $8$-dimensional space decomposes as
\begin{equation}
\mathcal H = \mathcal{H}_{\mathrm{triv}}^{\oplus 4} \oplus\mathcal{H}_{\mathrm{std}}^{\oplus2}.
\end{equation}
Thus, the logical register is 4-dimensional and therefore represents two logical qubits. A convenient basis for the code space is the set of symmetric Dicke states $\{\ket{000},\frac{1}{\sqrt{3}}(\ket{100}+\ket{010}+\ket{001}),\frac{1}{\sqrt{3}}(\ket{110}+\ket{101}+\ket{011}),\ket{111}\}$, and a basis for the standard isotypic subspace is given by $\{\frac{1}{\sqrt{2}}(\ket{100}-\ket{010}),\frac{1}{2}(\ket{100}+\ket{010}-2\ket{001}),\frac{1}{\sqrt{2}}(\ket{110}-\ket{101}),\frac{1}{2}(\ket{110}+\ket{101}-2\ket{011})\}$.

As an example, choose a logical codeword in the symmetric subspace; for simplicity, take the Dicke state
\begin{equation}
\ket{\psi}_L=\ket{000}.
\end{equation}
Consider first the error $E=X_0$, after which the logical state is
\begin{equation}
E\ket{\psi}_L = X_0\ket{000} = \ket{001}.
\end{equation}
We must measure the isotypic projector labels i.e. the syndrome. The probability of obtaining label $\lambda$ is
\begin{equation}
p_\lambda = \bra{\psi} E^\dagger \Pi_\lambda E \ket{\psi}_L.
\end{equation}
Substituting in the error $E=X_0$ produces probabilities
\begin{equation}
p_{\mathrm{triv}} = \frac{1}{3},\qquad
p_{\mathrm{std}} = \frac{2}{3},\qquad
p_{\mathrm{sgn}} = 0.
\end{equation}
Thus the single-qubit $X$ error on qubit $0$ maps the logical state partly out of the code-space: with probability $1/3$, the syndrome readout returns the trivial label and thus no detectable leakage, and with probability $2/3$, it returns the standard irrep label, a detectable symmetry-breaking error which projects the state into an isotypic component carrying the 2-dimensional standard representation. If the trivial representation is obtained, the state of the logical register after measurement is the $W$-state $\frac{1}{\sqrt{3}}(\ket{100}+\ket{010}+\ket{001})$, and so an undetectable error has occurred. If the standard representation is obtained, then the state after measurement is $\sqrt{\frac{2}{3}}\ket{001}-\frac{1}{\sqrt{6}}\ket{010}-\frac{1}{\sqrt{6}}\ket{100}$. The same computation can be repeated for any logical $\ket{\psi}_L\in\sym_{S_3}$ and any error operator $E$. This example shows that nonabelian codes beyond the Pauli group are not built to protect against Pauli-type errors. In order to understand the effectiveness of such a code, one must establish a basis for the unitary group on the physical Hilbert space which rotates logical states into single isotypic components. The standard notion of distance as the smallest number of nontrivial Pauli factors in a logical operation does not capture the utility of a nonabelian code, as we discussed in Section~\ref{sec:logical-ops}.

The primary advantage of the current code is that the code space is invariant under any permutation of the qubits. In architectures that lack all-to-all connectivity, such as superconducting transmons, a $\swapop$ overhead is introduced during transpilation. We can encode our logical qubits so that such overhead does not affect the logical state and is therefore unnecessary. In the event that qubits are erroneously permuted, the state is again unharmed due to the passive protection provided by the symmetric subspace. In addition to this passive protection, the isotypic syndrome extraction procedure provides active protection against other sources of error.

This code generalizes beautifully to a code on $n$-physical qubits. Unfortunately, such a code is unlikely to be useful because it requires the use of a QFT over the symmetric group. Although Beals \cite{beals1997} showed that the QFT over the symmetric group can be computed in time polynomial in $n$ using the Young-Yamanouchi basis, this scaling is exponential in the number of qubits required to represent a permutation. Hence, no efficient (i.e., polynomial in $\log(|S_n|)$) quantum circuit implementation is known. Nevertheless, it is a good exercise to understand how the general code would work in principle.

The action of the symmetric group on a tensor product is complemented by the collective unitary action of the unitary group $U(2)$ in the sense that they commute. That is, with $W(\sigma)\ket{\psi_1}\otimes\cdots\otimes\ket{\psi_n}=\ket{\psi_{\sigma^{-1}(n)}}\otimes\cdots\otimes\ket{\psi_{\sigma^{-1}(n)}}$ and with $V(U)\ket{\psi_1}\otimes\cdots\otimes\ket{\psi_n}=U\ket{\psi_1}\otimes\cdots\otimes U\ket{\psi_n}$, we have
\begin{equation}
    W(\sigma)V(U)=V(U)W(\sigma)
\end{equation}
for all $\sigma\in S_n$ and $U\in U(2)$. Schur-Weyl duality \cite{fulton2013,harrow2013} is a consequence of this fact and the double centralizer theorem, producing the canonical decomposition of the physical Hilbert space
\begin{equation}
    \mathcal{H}=\bigoplus_{\substack{\lambda\vdash n\\\ell(\lambda)\le2}}S^\lambda(\mathbb{C}^2)\otimes M^\lambda,
\end{equation}
where the direct sum is over the partitions of $n$ with length at most 2, $S^\lambda(\mathbb{C}^2)$ is the irreducible $U(2)$ module with label $\lambda$, and $M^\lambda$ is the corresponding irreducible $S_n$-module (Specht modules). For each $\lambda$, the subspace $S^\lambda(\mathbb{C}^2)\otimes M^\lambda$ is an isotypic component. The space $M^\lambda$ is the invariant subspace of the irrep of $S_n$ with label $\lambda$ and $S^\lambda(\mathbb{C}^2)$ is the multiplicity space. Now performing isotypic syndrome extraction amounts to finding the integer partition labeling the isotypic component to which the state belongs.

In the language of angular momentum theory \cite{sakurai2017modern}, this decomposition coincides with the standard addition of spin for $n$ spin-$\tfrac{1}{2}$ particles. Each partition $\lambda = (n-j,j)$ corresponds to a definite total angular momentum $J = \tfrac{n}{2} - j$, and the multiplicity space $S^{(n-j,j)}(\mathbb{C}^2)$ carries the $(2J+1)$-dimensional irreducible representation of $SU(2)$ associated with that total spin. The accompanying invariant space $M^{(n-j,j)}$ encodes the various ways the individual spins can be coupled to yield the same total $J$. Thus, measuring the isotypic component labeled by $\lambda$ is equivalent to performing a collective measurement of total spin. In this picture, isotypic syndrome extraction is nothing more than a symmetry-resolved spin measurement, partitioning the Hilbert space into distinct total-spin sectors that serve as natural error-syndrome subspaces.

\subsection{Alternative Symmetric Group Representation}\label{sec:alternative-sym}
To showcase how the code depends greatly on the chosen representation, consider again the case $G=S_3$, but this time let $\tau=(1\ 2)$ and $\eta=(1\ 2\ 3)$ denote the generators of $S_3$ and define a (faithful) 4-dimensional representation by $W(\tau)=\cnot$ and $W(\eta)=\cnot\cdot\swapop$. Since $\cnot$ and $\cnot\cdot\swapop$ are permutation matrices, every element of $W(S_3)$ is a permutation matrix. That is, each row and column contains a single 1 with all remaining elements equal to 0. Notice also that both $\tau$ and $\eta$ fix the first column/row and permute the remainder. The character of each $\sigma\in S_3$ is therefore given by the number of symbols fixed by $\sigma$ plus one. There are three transpositions, an identity element, and two $3$-cycles in $S_3$. Thus, the character table contains three elements with value 2, one element with value 4, and two elements with value 1. Now computing the number of copies of each irrep in $W$, we have
\begin{equation}
    \langle\chi_W,\chi_{\textnormal{triv}}\rangle=2
\end{equation}
\begin{equation}
    \langle\chi_W,\chi_{\textnormal{sgn}}\rangle=0
\end{equation}
\begin{equation}
    \langle\chi_W,\chi_{\textnormal{std}}\rangle=1.
\end{equation}
Thus, our representation reduces to $W=\rho_{\textnormal{triv}}^{\oplus 2}\oplus\rho_{\textnormal{std}}$.

The code space, which is the space of all $S_3$-symmetric states with respect to the representation $W$, is 2-dimensional and therefore represents a single qubit. A convenient basis for this space is $\{\ket{00},\frac{1}{\sqrt{3}}(\ket{01}+\ket{10}+\ket{11})\}$. For simplicity, consider a code word $\ket{00}$. Applying an $X_1$ error produces the state $\ket{10}$. Now performing isotypic syndrome extraction, we find that the probabilities corresponding to each irrep are
\begin{equation}
    p_{\textnormal{triv}}=\frac13,
    \qquad
    p_{\textnormal{std}}=\frac23,\qquad p_{\textnormal{sgn}}=0.
\end{equation}
Like in the permutation code, this code is therefore not set up to protect specifically against Pauli-$X$ errors.

Our symmetric codes are best prepared to correct errors which rotate a logical state into precisely one nontrivial isotypic subspace so that a recovery operation can be applied depending on this label, correcting the error every time. In the case of the current code, there is only one non-trivial isotypic subspace, and it has a basis $\{\ket{w_1}=\frac{1}{\sqrt{2}}(\ket{01}-\ket{11}),\ket{w_2}=\frac{1}{\sqrt{2}}(\ket{10}-\ket{11})\}$. An example of an operator that this code can correct is $U=Z_1X_0\cnot_{01}H_0$. Indeed, this operator maps the code word $\alpha\ket{00}+\frac{\beta}{\sqrt{3}}(\ket{01}+\ket{10}+\ket{11})$ to the state $\frac{\alpha}{\sqrt{2}}(\ket{10}-\ket{01})+\frac{\beta}{\sqrt{6}}(\ket{10}+\ket{01}-2\ket{11})=\alpha(\ket{w_2}-\ket{w_1})+\frac{\beta}{\sqrt{3}}(\ket{w_1}+\ket{w_2})$, which is equivalent to the state
\begin{equation}
    \left(\frac{\beta}{\sqrt{3}}-\alpha\right)\ket{w_1}+\left(\alpha+\frac{\beta}{\sqrt{3}}\right)\ket{w_2},
\end{equation}
and is clearly in the non-trivial isotypic component. Thus, if we wish to correct the error $U$ specifically, we would decode the non-trivial $\lambda$ syndrome by applying $U^\dagger$.

\subsection{One Logical Qubit and the Dihedral Group}\label{sec:one-logical}
As pointed out in Section~\ref{sec:two-logical}, in the context of fault tolerance, we would prefer to protect logical qubits with independent codes. Let us therefore build a technique for constructing such codes. Let $G$ be a finite group that acts on a finite set $\mathscr{A}$. There is a permutation representation induced by this action. Indeed, the action on $\mathscr{A}$ induces a unitary representation $W$ on the formal Hilbert space $\mathbb{C}[\mathscr{A}]=\{e_a:a\in \mathscr{A}\}$ defined by $W(g) e_{a}=e_{g\cdot a}$. Since $W$ is completely reducible, $\mathbb{C}[\mathscr{A}]$ decomposes as
\begin{equation}
    \mathbb{C}[\mathscr{A}]=\bigoplus_{\lambda}\mathcal{H}_\lambda^{\oplus m_\lambda},
\end{equation}
where $\mathcal{H}_\lambda$ is the invariant subspace of the irrep with label $\lambda$ and $m_\lambda$ is the multiplicity of the same irrep in $W$. The multiplicity of the trivial representation determines the number of logical qubits and is given by the character orthogonality relation
\begin{equation}
    m_{\textnormal{triv}}=\langle\chi_W,\chi_{\textnormal{triv}}\rangle=\frac{1}{|G|}\sum_{g\in G}\chi_W(g).
\end{equation}
The character of the $W$ representation is given by the number of $a\in \mathscr{A}$ fixed by $W(g)$. Thus,
\begin{equation}
    m_{\textnormal{triv}}=\frac{1}{|G|}\sum_{g\in G}\lvert\{a\in \mathscr{A}:g\cdot a=a\}\rvert,
\end{equation}
but by Burnside's lemma \cite{burnside1897}, this is precisely the number of orbits in $\mathscr{A}/G$, the number of equivalence classes of elements of $\mathscr{A}$ modulo the action of $G$. It follows that
\begin{equation}
    m_{\textnormal{triv}}=\lvert \mathscr{A}/G\rvert.
\end{equation}

Each orbit corresponds to one invariant vector in $\mathbb{C}[\mathscr{A}]$, and the symmetric subspace $\sym_G$ is spanned by the uniform superpositions
\begin{equation}
    \ket{\psi_i}=\frac{1}{\sqrt{|\mathcal{O}_i|}}\sum_{a\in\mathcal{O}_i}\ket{a},
\end{equation}
where each $\mathcal{O}_i$ is an orbit belonging to $\mathscr{A}/G$. Thus, in order to build a code for a single logical qubit, it suffices to choose $\mathscr{A}$ and $G$ so that there are only two orbits. The basis for the symmetric subspace is then given by the uniform superpositions over the elements of each orbit. One way to accomplish this is by considering the action of $G$ on a disjoint union of sets for which it acts transitively. Indeed, if $G$ acts transitively on $\mathscr{A}_1$ via some representation $W_1$ and on $\mathscr{A}_2$ via some representation $W_2$, then there is a natural representation $W:G\to \mathcal{U}(\mathbb{C}[\mathscr{A}_1\sqcup \mathscr{A}_2])$, and the quotient defined by this action has size $\lvert \mathscr{A}_1\sqcup \mathscr{A}_2/G\rvert=2$ with orbits $\mathscr{A}_1$ and $\mathscr{A}_2$.

In Section~\ref{sec:ex_S3}, we remarked that the symmetric group code for large $n$ is intractable due to the complexity of the corresponding QFT. This is largely due to the exponential blow up of the irrep dimensions as $n$ grows. As the dihedral group $D_n$ is isomorphic to a semidirect product of cyclic groups, the corresponding QFT for this group is comparatively easy to implement. This line of reasoning indicates that it might be worthwhile to study an example involving this group.

To be precise, the dihedral group $D_n$ can be written as a semidirect product
\begin{equation}
    D_n \cong \mathbb{Z}_n \rtimes \mathbb{Z}_2,
\end{equation}
where $\mathbb{Z}_n$ represents the $n$ rotations and $\mathbb{Z}_2$ the reflection. This group has $2n$ elements and its irreducible representations consist of two one-dimensional irreps when $n$ is odd (four when $n$ is even), and $(n-1)/2$ two-dimensional irreps when $n$ is odd (and $n/2 - 1$ when $n$ is even). Thus, all irreps of $D_n$ have dimension at most two, in contrast with the factorial growth of the irrep dimensions of $S_n$.

An efficient QFT over $D_n$ exploits this simple structure. It proceeds by first performing a QFT over the abelian subgroup $\mathbb{Z}_n$, followed by a short sequence of controlled operations that yields an overall circuit complexity of $O(\log^2(n))$ gates \cite{hoyer1997,puschel1999,childs2010}. The efficiency of this construction ultimately derives from the fact that both the subgroup $\mathbb{Z}_n$ and the quotient $\mathbb{Z}_2$ are abelian, allowing the semidirect product structure to be exploited algorithmically. We work out the case of $D_{2^m}$ explicitly in Section~\ref{sec:qft-dn}.

The relative simplicity of the dihedral group makes it a natural next step for constructing symmetry-based codes beyond the stabilizer codes. A dihedral code would encode logical states into invariant subspaces corresponding to the irreducible representations of $D_n$, with isotypic syndrome extraction implemented using the efficient QFT over $D_n$. Such codes could in principle retain some of the desirable symmetry properties of the symmetric-group construction while remaining physically realizable for larger system sizes.

There is a natural choice for $\mathscr{A}_1$ and $\mathscr{A}_2$ when $G=D_n$. We take $\mathscr{A}_1=\{\ket{i}\ket{0}\}_{i=0}^{n-1}$ to be the labels of the vertices which get permuted by the elements of $D_n$, and we take $\mathscr{A}_2=\{\ket{i}\ket{1}\}_{i=0}^{n-1}$ to be the labels of the edges which are also permuted by the elements of $D_n$. In fact, $D_n$ acts transitively on both sets via $W_j(\sigma)\ket{i}\ket{j}=\ket{\sigma(i)}\ket{j}$, where $\sigma$ is the permutation associated to an element of $D_n$. Thus, there is a unitary representation $W:D_n\to\mathcal{U}(\mathbb{C}[\mathscr{A}_1\sqcup \mathscr{A}_2])$, and this representation defines a code with a two dimensional code space with basis $\{\frac{1}{\sqrt{n}}(\ket{0}+\cdots +\ket{n-1})\ket{0},\frac{1}{\sqrt{n}}(\ket{0}+\cdots+\ket{n-1})\ket{1}\}$. 

Let us count the number of available syndromes. Since $W_0$ and $W_1$ are isomorphic representations, it suffices to decompose $W_0$ into irreducible representations. We may further simplify the problem by noting that $W_0$ is isomorphic to the representation $V$ obtained by dropping the extra $\ket{0}$ factor, so that vertices of the $n$-gon are labeled by $\ket{i}$ and $V(\sigma)\ket{i}=\ket{\sigma(i)}$. Observe that
\begin{equation}
    V(\sigma)\frac{1}{\sqrt{n}}\sum_{i=0}^{n-1}\ket{i}=\frac{1}{\sqrt{n}}\sum_{i=0}^{n-1}\ket{i}
\end{equation}
for all $\sigma\in D_n$, where we are again identifying $D_n$ with its natural image in $S_n$. The one dimensional subspace spanned by $\frac{1}{\sqrt{n}}\sum_{i=0}^{n-1}\ket{i}$ is the trivial isotypic component (the symmetric subspace), and it remains to identify a decomposition of the remaining $(n-1)$-dimensional subspace. To this end, let us make a change of basis to the Fourier basis defined by
\begin{equation}
    \ket{\widehat{k}}=\frac{1}{\sqrt{n}}\sum_{j=0}^{n-1}\omega^{jk}\ket{j},
\end{equation}
where $\omega=e^{2\pi i/n}$ is the principal $n$-th root of unity. For $k=0$, we recover the state spanning the invariant subspace of the trivial representation. For $1\le k\le n-1$, consider the action of a rotation and a flip. For the rotation $r$, we have
\begin{equation}
    V(r)\ket{\widehat{k}}=\omega^k\ket{\widehat{k}},
\end{equation}
and for a flip $f$, we see that
\begin{equation}
    V(f)\ket{\widehat{k}}=\ket{\widehat{n-k}}.
\end{equation}
It follows that as long as $k\ne n-k$, the invariant subspace spanned by $\{\ket{\widehat{k}},\ket{\widehat{n-k}}\}$ is 2-dimensional and (in the Fourier basis) carries the 2-dimensional irreducible representation defined by
\begin{equation}\label{eq:dihedral-irreps}
    r\mapsto\begin{pmatrix}
        \omega^k&0\\0&\omega^{-k}
    \end{pmatrix},\;\;
    f\mapsto\begin{pmatrix}
        0&1\\1&0
    \end{pmatrix}.
\end{equation}
Now, $k=n-k$ can only hold when $n=2k$ is even. Thus, for odd $n$, $V$ decomposes as one copy of the trivial irrep plus $\frac{n-1}{2}$ distinct 2-dimensional irreps for a total of $\frac{n+1}{2}$ isotypic components. When $n$ is even, we pair $\ket{\widehat{k}}$ with $\ket{\widehat{n-k}}$ except for the case $k=\frac{n}{2}$, where $k=n-k$. In this case, $\ket{\widehat{k}}=\ket{\widehat{n-k}}=\ket{\widehat{n/2}}=\frac{1}{\sqrt{n}}\sum_{j=0}^{n-1}(-1)^{j}\ket{j}$, and it follows that a rotation multiplies the vector by $-1$ and a flip leaves the vector invariant. Thus, $\ket{\widehat{n/2}}$ spans a 1-dimensional invariant subspace carrying the irreducible representation defined by $r\mapsto-1$, $f\mapsto1$. It follows that, for even $n$, $V$ decomposes as one copy of the trivial irrep plus one copy of the above 1-dimensional irrep, plus $\frac{n-2}{2}$ distinct 2-dimensional irreps for a total of $\frac{n+2}{2}$ isotypic components.

Since $W$ is isomorphic to two copies of $V$, $W$ has the same number of isotypic components as $V$. As there is a syndrome associated to every isotypic component, the code defined by $W$ has $\frac{n+1}{2}$ available syndromes when $n$ is odd and $\frac{n+2}{2}$ when $n$ is even. Since one of these syndromes is associated to the code space, the code distinguishes
\begin{equation}
    \begin{cases}
        \frac{n-1}{2},&\textnormal{n odd}\\
        \frac{n}{2},&\textnormal{n even}
    \end{cases}
\end{equation}
nontrivial syndromes. In order to obtain $n$ basis states, we need $\lceil\log_2(n)\rceil$ qubits, and so the rate of the code is $R=\frac{1}{\lceil\log_2(n)\rceil}$.

Recall that the centralizer contains logical operations which do not permute the nontrivial isotypic components. The structure of this subgroup is obtained easily from Schur's lemma. Since every irrep in the decomposition of a single copy of the permutation representation appears with multiplicity one, taking two copies means that each irrep now appears with multiplicity two. By Schur's lemma, the commutant of $W(D_n)$ on the isotypic component labeled by $\lambda$ is therefore isomorphic to the full unitary group $U(2)$ acting on the multiplicity space. Denoting the corresponding unitary block by $U_\lambda$, we see that the matrix representation of any element of the centralizer (in the Fourier basis) has the block diagonal form $\textnormal{diag}(U_0,U_{n/2},U_1\otimes I_2,\ldots,U_{(n-1)/2}\otimes I_2)$ for $n$ even and $\textnormal{diag}(U_0,U_1\otimes I_2,\ldots,U_{n/2}\otimes I_2)$ for $n$ odd. Equivalently, the centralizer is isomorphic to $U(2)\times\cdots\times U(2)$, one copy of $U(2)$ for each distinct irrep of $D_n$ appearing in the decomposition. For even $n$, note that the one-dimensional representations corresponding to $k=0$ and $k=n/2$ each contribute their own $U(2)$ factor, while the remaining 2-dimensional irreps each contribute a $U(2)$ acting on their twofold multiplicity space.

The full normalizer is slightly more complicated, but its elements can be classified modulo the centralizer. Let $N_{U(\mathcal{H})}(W(D_n))$ be the normalizer and note that the kernel of the map $N_{U(\mathcal{H})}(W(D_n))\to\textnormal{Aut}(W(D_n))$ defined by $g\mapsto g(\cdot)g^\dagger$ is the centralizer $C_{U(\mathcal{H})}(W(D_n))$. It follows from the first isomorphism theorem that
\begin{equation}
    N_{U(\mathcal{H})}(W(D_n))/C_{U(\mathcal{H})}(W(D_n))\cong \textnormal{Aut}(W(D_n)).
\end{equation}
We have already determined the structure of the centralizer, and the elements of the automorphism group of $D_n$ have the form $r\mapsto r^a$ and $f\mapsto r^bf$ for some $a,b\in\mathbb{Z}_n$ with $a\ne0$. Since $W$ is faithful, we may lift this structure to the image of $D_n$ so that the automorphisms satisfy $W(r)\mapsto (W(r))^a$ and $W(f)\mapsto (W(r))^bW(f)$. Define a permutation matrix $U_{a,b}$ by its action on the Fourier basis:
\begin{equation}
    U_{a,b}\ket{\widehat{k},\alpha}=\omega^{bk}\ket{\widehat{ak},\alpha},
\end{equation}
where $\alpha=0,1$ is labeling the copy of the Fourier basis state (recall there are two). We claim that conjugation by the permutation matrix $U_{a^{-1},2^{-1}ba^{-1}}$ implements the above automorphism. Indeed, for a rotation $r$, we have
\begin{align}
    &U_{a^{-1},2^{-1}ba^{-1}}W(r)U_{a^{-1},2^{-1}ba^{-1}}^\dagger\ket{\widehat{k},\alpha}\\
    &=U_{a^{-1},2^{-1}ba^{-1}}W(r)\omega^{-2^{-1}bk}\ket{\widehat{ak},\alpha}\\
    &=\omega^{-2^{-1}bk}U_{a^{-1},2^{-1}ba^{-1}}\omega^{ak}\ket{\widehat{ak},\alpha}\\
    &=\omega^{-2^{-1}bk}\omega^{ak}\omega^{2^{-1}bk}\ket{\widehat{k},\alpha}\\
    &=\omega^{ak}\ket{\widehat{k}}=W(r^a)\ket{\widehat{k},\alpha},
\end{align}
and for a flip $f$, we have
\begin{align}
    &U_{a^{-1},2^{-1}ba^{-1}}W(f)U_{a^{-1},2^{-1}ba^{-1}}^\dagger\ket{\widehat{k},\alpha}\\
    &=U_{a^{-1},2^{-1}ba^{-1}}W(f)\omega^{-2^{-1}bk}\ket{\widehat{ak},\alpha}\\
    &=\omega^{-2^{-1}bk}U_{a^{-1},2^{-1}ba^{-1}}\ket{-\widehat{ak},\alpha}\\
    &=\omega^{-2^{-1}bk}\omega^{-2^{-1}bk}\ket{-\widehat{k},\alpha}\\
    &=\omega^{-bk}\ket{-\widehat{k},\alpha}=W(r^bf)\ket{\widehat{k},\alpha}.
\end{align}
Thus, elements of the normalizer in the Fourier basis have the form $D\cdot U^{\oplus 2}_{a,b}$, where $D$ is a block-diagonal matrix in the centralizer. (Recall that the exact form depends on the parity of $n$) 

Since the logical operations are the elements of the normalizer, we see that the action of a logical operation is to permute the isotypic component blocks and then apply a block diagonal unitary operation with one block for each isotypic component. When the logical state is uncorrupted, this permutation of the isotypic components does nothing, and we are left with an arbitrary single qubit unitary operator on the code space. Thus, the logical qubit in this construction admits a full set of logical operations, including non-Clifford rotations. However, these are not necessarily transversal operations. Their physical realization depends on the underlying hardware and on which gates preserve the encoded subspace. This observation therefore does not violate the Eastin–Knill theorem \cite{eastin2009}; it simply reflects that, at the logical level, the code supports the entire $U(2)$ algebra of transformations. 

Unfortunately, this feature is also a bug. Since $U(2)$ is a continuous group, there are infinitely many errors which can masquerade as a logical operation, thereby going undetected. 
%To be clear, this is also the case for a stabilizer code when we consider the ambient space that the cyclic group sits inside of to be the entire unitary group instead of the Pauli group. 
To be clear, this is also potentially the case for stabilizer codes. Consider letting the ambient space that cyclic group sits inside of be the entire unitary group---a natural extension of the usual Pauli group construction. Then the same problem arises.
The problem is that we have not restricted our attention to a subgroup whose group algebra spans the full operator algebra on the Hilbert space; that is, we have not taken advantage of the discretization of errors. We should therefore look for a subgroup of the unitary group analogous to the Pauli group which the image of the dihedral group naturally embeds in. The group algebra generated by this subgroup should span the full operator algebra on the code space, so that any error operator can be expanded as a linear combination of group elements. 

Let $\mathcal{Q}$ denote such a subgroup. The logical operations are then given by the normalizer $N_{\mathcal{Q}}(W(D_n))$. For the sake of clarity, let us choose $n=2^m$ so that $W(D_n)$ acts on $m+1$ qubits (recall that $W$ is two copies of the basis permutation representation). Using the isomorphism~\footnote{This isomorphism follows from $\dim(\mathbb{C}^{2^m}\oplus\mathbb{C}^{2^m})=2^m+2^m=2\cdot 2^m=\dim(\mathbb{C}^{2^m}\otimes\mathbb{C}^2)$.} 
$\mathbb{C}^{2^m}\oplus\mathbb{C}^{2^m}\cong\mathbb{C}^{2^m}\otimes\mathbb{C}^2$, we label the basis states $\ket{i,\alpha}$, where $i=0,\ldots,2^{m}-1$ and $\alpha=0,1$. Setting $\mathcal{Q}=\langle \pm i,I\otimes X,I\otimes Z, W(D_n)\rangle$, where the identity acts on $\mathbb{C}^{2^m}$ and the Pauli-$X$ and $Z$ gates act on $\mathbb{C}^2$, we see that the normalizer is
\begin{equation}
    N_{\mathcal{Q}}(W(D_n))=\mathcal{Q}
\end{equation}
since the Pauli operations act on the copy labeling register, while the representation $W$ acts on $\mathbb{C}^{2^m}$. It follows that the non-trivial logical operations are isomorphic to the single qubit Pauli group, with $I\otimes X$ taking the role of the logical $X_L$ operation and $I\otimes Z$ taking the role of the logical $Z_L$ operation. However, $\mathcal{Q}$ does not span the unitary group. To fix this, we add a diagonal matrix of phases to the list of generators. Define
\begin{equation}
    D=\textnormal{diag}(e^{i\theta_0},\ldots,e^{i\theta_{2^m-1}}),
\end{equation}
where the phases $\theta_j$ are pairwise distinct---for example, $\theta_j=\frac{2\pi j}{2^m+1}$. We claim that the addition of $D$ as a generator is all that is needed to span the full unitary group.

\begin{proposition}
    Any element of $U(\mathbb{C}^{2^m}\otimes\mathbb{C}^2)$ can be written as a linear combination of elements of the group $\mathcal{Q}=\langle \pm i,I\otimes X,I\otimes Z, W(D_n), D\rangle$.
\end{proposition}
\begin{proof}
    The copy of the single qubit Pauli group generated by $\pm i,I\otimes X,I\otimes Z$ certainly spans $U(\mathbb{C}^2)$, so it remains to show that $W(D_n)$ and $D$ together span $U(\mathbb{C}^{2^m})$. To this end, recall that there is a unique polynomial $p(x)$ of degree at most $2^m-1$ such that $p(x_j)=y_j$ for pairs $(x_0,y_0),\ldots, (x_{2^m-1},y_{2^m-1})$. Indeed, this polynomial is constructed by Lagrange interpolation. Explicitly, we have
    \begin{equation}
        p(x)=\sum_{j=0}^{2^m-1}y_jL_j(x),
    \end{equation}
    where $L_j(x):=\prod_{k\ne j}\frac{x-x_k}{x_j-x_k}$ denotes the Lagrange basis polynomials. 
    
    Note that $L_j(x_k)=\delta_{jk}$ and set $y_j=1$ for some $j$ with all other $y_k=0$ for $k\ne j$. We take $x_j=e^{i\theta_j}$. Then $p(x)=L_j(x)$. Since $D$ is diagonal, $p(D)=\textnormal{diag}(p(e^{i\theta_0}),\ldots,p(e^{i\theta_{2^m-1}}))$, and so we have
    \begin{align}
        p(D)&=\textnormal{diag}(L_j(e^{i\theta_0}),\ldots,L_j(e^{i\theta_{2^m-1}}))\\
        &=\textnormal{diag}(\delta_{j,0},\ldots,\delta_{j,2^m-1})\\
        &=\ket{j}\!\!\bra{j}.
    \end{align}
    Since $D_n$ acts transitively on the set $\{0,\ldots,2^m-1\}$, there is a permutation matrix $P_{ji}$ for any $i,j$ such that $P_{ji}\ket{j}=\ket{i}$. Thus, we obtain the $\ket{i}\!\!\bra{j}$ by applying this operation. Since $\ket{i}\!\!\bra{j}$ is a basis for $M_{2^m}(\mathbb{C})$, we have that $W(D_n)$ and $D$ together span $U(\mathbb{C}^{2^m})$.
\end{proof}

With $\mathcal{Q}=\langle \pm i,I\otimes X,I\otimes Z, W(D_n), D\rangle$, the normalizer is no longer all of $\mathcal{Q}$, but it does still contain the subgroup $\langle \pm i,I\otimes X,I\otimes Z, W(D_n)\rangle$, so that the entire logical Pauli group can be implemented. The detectable errors are given by the cosets of the normalizer $N_{\mathcal{Q}}(W(D_n))$ in $\mathcal{Q}$. This construction puts the dihedral code on the same playing field as the ordinary stabilizer code. The important difference is that a stabilizer code protects against an abelian subgroup of the Pauli group passively, while the dihedral code passively protects against certain permutations of the basis states.

\section{Circuit Implementations for some Special Cases}\label{sec:qft}
%we can always go back and make things consistent later
\subsection{Discussion of $D_3$ Case}
While the mathematics developed in previous sections is sufficiently complete to construct a code, it is helpful to see explicitly how a circuit-based implementation is constructed. For this reason, this section will focus on one illustrative example, the dihedral group of the triangle, $D_3$, which is isomorphic to $S_3$. The goal here is to demonstrate a step-by-step procedure by focusing on a specific group implementation, though the methods discussed generalize to any group.

First, consider the tensor permutation representation of $D_3$. If we label a three-qubit system sequentially as $\ket{q_1q_2q_3}$ then our unitary representation is given by 
\begin{align}
    e &= 1 \\
    f &= \swapop_{13} \\
    r &= \swapop_{23}\cdot \swapop_{12} \\
    fr &= \swapop_{13} \cdot \swapop_{23}\cdot \swapop_{12}\\
    r^2 & = \swapop_{12}\cdot \swapop_{23} \\
    f r^2 & = \swapop_{13} \cdot \swapop_{12}\cdot \swapop_{23} \, ,
\end{align}
where $f$ and $r$ denote the flip and rotation operations of the dihredral group, respectively. The projection $\Pi_G$ onto the symmetric subspace $\sym_G$ given by our usual $G$-Bose symmetry test is 
\begin{equation}
    \Pi_G = \frac{1}{3} \begin{pmatrix}
        3 & 0 & 0 & 0 & 0 & 0 & 0 & 0 \\
        0 & 1 & 1 & 0 & 1 & 0 & 0 & 0 \\
        0 & 1 & 1 & 0 & 1 & 0 & 0 & 0 \\
        0 & 0 & 0 & 1 & 0 & 1 & 1 & 0 \\
        0 & 1 & 1 & 0 & 1 & 0 & 0 & 0 \\
        0 & 0 & 0 & 1 & 0 & 1 & 1 & 0 \\
        0 & 0 & 0& 1 & 0 & 1 & 1 & 0 \\
        0 & 0 & 0 & 0 & 0 & 0 & 0 & 3 \\
    \end{pmatrix} \, .
\end{equation} %maybe not necessary to include here

The next step concerns specifying the ancillary register which implements the group action on the input state. Let
\begin{equation}
    \ket{+_G}\coloneqq \frac{1}{\sqrt{|G|}}\sum_{g\in G} \ket{g}
\end{equation}
denote the desired control state for a generic group $G$. This ancilla realizes the states given in Eqs.~\eqref{eq:ancilla+input1}-\eqref{eq:ancilla+input2}. For $D_3$, we explicitly assign the group elements as 
\begin{align}
    \ket{e} &= \ket{000} \, ,\\
    \ket{f} &= \ket{001} \, ,\\
    \ket{r} &= \ket{010} \, ,\\
    \ket{rf} &= \ket{011} \, ,\\
    \ket{r^2} &= \ket{100} \, ,\\
    \ket{r^2 f} &= \ket{101} \, ,
\end{align}
such that the ancilla state becomes
\begin{equation}\label{eq:D3_superposition}
    \ket{+_{D_3}} = \frac{1}{\sqrt{6}}\left( \ket{000} +\ket{001}+ \ket{010} + \ket{011} +\ket{100} +\ket{101}\right) \, . 
\end{equation}
Generally, each control state can be written as $\ket{r^k}\otimes\ket{f^\alpha}$ where $k\in \mathbb{Z}_n$ and $\alpha \in \mathbb{Z}_2$.

Constructing the superposition in Eq.~\eqref{eq:D3_superposition} is fairly straightforward. The control state can be broken into a tensor product as $\frac{1}{\sqrt{3}}\left (\ket{00} + \ket{01} + \ket{10} \right) \otimes \frac{1}{\sqrt{2}}\left (\ket{0} + \ket{1} \right )$. The second term is generated via a single Hadamard gate, leaving only the first remaining. For convenience, let $U_i(\theta,\phi,\lambda)$ denote a unitary matrix acting on the $i$-th qubit defined by three Euler angles such that
\begin{equation}\label{eq:gen-unitary}
U_{i}(\theta,\phi,\lambda) \coloneqq
    \begin{pmatrix}
        \cos{\frac{\theta}{2}} & -e^{-i\lambda}\sin{\frac{\theta}{2}} \\
        e^{i\phi}\sin{\frac{\theta}{2}} & e^{i(\phi+\lambda)}\cos{\frac{\theta}{2}}
    \end{pmatrix} \, .
\end{equation}
Then the state $\frac{1}{\sqrt{3}}\left (\ket{00} + \ket{01} + \ket{10} \right)$ can be generated by  acting on $\ket{00}$ with the unitary
$
    U_0(\frac{\pi}{4},-\pi,\pi)\ \cnot_{1\to0} \ U_0(\frac{3\pi}{4},-2\pi,\frac{-\pi}{2})\ U_1(\gamma,0,0) \, ,
$ where $\gamma = 2 \cos^{-1} \left ( \frac{{\csc{\frac{\pi}{8}}}\sec{\frac{\pi}{8}}}{2 \sqrt{3}} \right) = 1.23096$.
The full construction is shown in Figure~\ref{fig:d3_superposition}. More generally, a uniform superposition of $n$ computational basis states can be generated with a circuit depth of $O(\log_2(n))$ \cite{shukla2024}.

Let us now construct the QFT for the dihedral group according to Eq.~\eqref{eq:group-QFT}. Observe that $\operatorname{QFT}_{G}$ always acts on the \textit{ancillary} register, which means that what matters is the mapping between the ancilla representation and the irreps of the group, \textit{not} the representation on the logical register. For this example, we label the Fourier basis $\{ \ket{\lambda,i,j}\}$ using the convention that $\lambda=1$ for the trivial representation, $-1$ for the sign representation, and $2$ for the two-dimensional standard representation. This convention denotes the output basis for the operator, where the input basis shall be the group elements in order of their computational basis state.

The Fourier coefficient of the state $\ket{\lambda,i,j}$ is given by
\begin{equation}
    c_{\lambda,i,j} (g) \coloneqq \frac{\sqrt{d_\lambda}}{\sqrt{|G|}} \bra{i}\rho_\lambda(g)\ket{j} = \frac{\sqrt{d_\lambda}}{\sqrt{|G|}}[\rho_\lambda(g)]_{i,j}\, ,
\end{equation}
where the matrix element on the right-hand side gives the $i,j$-th entry of the \textit{normalized} irreducible representation $\rho_\lambda$. For the trivial representation, $\rho_1(g)$ is $1$ for all group elements. For the sign representation, $\rho_{-1}(g) =\operatorname{sgn}(g)$. However, there is some ambiguity in the matrix representation of the dihedral group; we resolve this by restricting to unitary representations, as otherwise the resultant $\operatorname{QFT}_{D_3}$ will not be unitary. As such, we use the complex-valued representation given by
\begin{align}
    \rho_{2}(e) &= \begin{pmatrix}
        1 & 0 \\
        0 & 1
    \end{pmatrix}\, ,  \\
    \rho_{2} (f) &= \begin{pmatrix}
        0 & 1 \\
        1 & 0 
    \end{pmatrix} \, ,\\
    \rho_{2}(r) &= \begin{pmatrix}
        e^{2\pi i/3} & 0 \\
        0 & e^{-2\pi i/3}
    \end{pmatrix} \, ,\\
    \rho_{2}(rf) &= \begin{pmatrix}
        0 & e^{-2\pi i/3} \\
        e^{2\pi i/3} & 0
    \end{pmatrix} \, ,\\
    \rho_{2}(r^2) &= \begin{pmatrix}
        e^{-2\pi i/3} & 0 \\
        0 & e^{2\pi i/3}
    \end{pmatrix} \, ,\\
    \rho_{2}(rf) &= \begin{pmatrix}
        0 & e^{2\pi i/3} \\
        e^{-2\pi i/3} & 0
    \end{pmatrix} \, .
\end{align}

\begin{figure}[t!]
    \centering
    \includegraphics[width=\columnwidth]{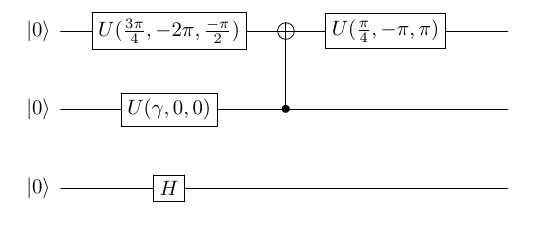}
    \caption{State creation circuit used to generate $\ket{+_{D_3}}$ from single-qubit unitaries of the form \eqref{eq:gen-unitary} and a single CNOT gate. In the circuit, $\gamma = 2 \cos^{-1} \left ( \frac{{\csc{\frac{\pi}{8}}}\sec{\frac{\pi}{8}}}{2 \sqrt{3}} \right) = 1.23096$. }
    \label{fig:d3_superposition}
\end{figure}

Now all that is left is to define the operator itself. Note that the output basis of the $\operatorname{QFT}_{D_3}$ is given by $\{\ket{1,1,1},\ket{-1,1,1},\ket{2,1,1},\ket{2,1,2},\ket{2,2,1},\ket{2,2,2}\}$, in that order. This is a 6-dimensional basis, just as there are six elements in the group; however, to force the operation to be unitary on qubits, we pad this basis with two auxiliary orthonormal vectors to allow it to act on a three-qubit register. Since the states $\ket{110}$ and $\ket{111}$ do not occur in $\ket{+_{D_3}}$, we are safe to let these two states map to the auxiliary vectors without concern for how the calculation may be affected.

The $\operatorname{QFT}_{D_3}$ can then be written as
\begin{align*}
    \left(
\begin{smallmatrix}
 c_{1,1,1}(e) &  c_{1,1,1}(f) & c_{1,1,1}(r) & c_{1,1,1}(rf) & c_{1,1,1}(r^2) & c_{1,1,1}(r^2f) & 0 & 0 \\
 c_{\text{-}1,1,1}(e) & c_{\text{-}1,1,1}(f) & c_{\text{-}1,1,1}(r) & c_{\text{-}1,1,1}(rf) & c_{\text{-}1,1,1}(r^2) & c_{\text{-}1,1,1}(r^2f) & 0 & 0 \\
 c_{2,1,1}(e) & c_{2,1,1}(f)  & c_{2,1,1}(r) & c_{2,1,1}(rf)  & c_{2,1,1}(r^2)  & c_{2,1,1}(r^2f)  & 0 & 0 \\
 c_{2,1,2}(e)  & c_{2,1,2}(f)  & c_{2,1,2}(r) & c_{2,1,2}(rf) & c_{2,1,2}(r^2) & c_{2,1,2}(r^2f)  & 0 & 0 \\
c_{2,2,1}(e)  & c_{2,2,1}(f)  & c_{2,2,1}(r) & c_{2,2,1}(rf) & c_{2,2,1}(r^2) & c_{2,2,1}(r^2f)& 0 & 0 \\
c_{2,2,2}(e)  & c_{2,2,2}(f)  & c_{2,2,2}(r) & c_{2,2,2}(rf) & c_{2,2,2}(r^2) & c_{2,2,2} (r^2f) & 0 & 0 \\
 0 & 0 & 0 & 0 & 0 & 0 & 1 & 0 \\
 0 & 0 & 0 & 0 & 0 & 0 & 0 &  1 \\
\end{smallmatrix}
\right) \, ,
\end{align*}
or, explicitly,
\begin{equation}\label{eq:qftd3-numbers}
    QFT_{D_3} =\left(
\begin{smallmatrix}
 \frac{1}{\sqrt{6}} & \frac{1}{\sqrt{6}} & \frac{1}{\sqrt{6}} & \frac{1}{\sqrt{6}} & \frac{1}{\sqrt{6}} & \frac{1}{\sqrt{6}} & 0 & 0 \\
 \frac{1}{\sqrt{6}} & -\frac{1}{\sqrt{6}} & \frac{1}{\sqrt{6}} & -\frac{1}{\sqrt{6}} & \frac{1}{\sqrt{6}} & -\frac{1}{\sqrt{6}} & 0 & 0 \\
 \frac{1}{\sqrt{3}} & 0 & \frac{e^{\frac{2 i \pi }{3}}}{\sqrt{3}} & 0 & \frac{e^{-\frac{2 i \pi }{3}}}{\sqrt{3}} & 0 & 0 & 0 \\
 0 & \frac{1}{\sqrt{3}} & 0 & \frac{e^{-\frac{2 i \pi }{3}}}{\sqrt{3}} & 0 & \frac{e^{\frac{2 i \pi }{3}}}{\sqrt{3}} & 0 & 0 \\
 0 & \frac{1}{\sqrt{3}} & 0 & \frac{e^{\frac{2 i \pi }{3}}}{\sqrt{3}} & 0 & \frac{e^{-\frac{2 i \pi }{3}}}{\sqrt{3}} & 0 & 0 \\
 \frac{1}{\sqrt{3}} & 0 & \frac{e^{-\frac{2 i \pi }{3}}}{\sqrt{3}} & 0 & \frac{e^{\frac{2 i \pi }{3}}}{\sqrt{3}} & 0 & 0 & 0 \\
 0 & 0 & 0 & 0 & 0 & 0 & 1 & 0 \\
 0 & 0 & 0 & 0 & 0 & 0 & 0 & 1 \\
\end{smallmatrix}
\right) \, .
\end{equation}
which is a unitary matrix as desired. Therefore, Eq.~\eqref{eq:qftd3-numbers} can be implemented on a quantum computer via a decomposition into elementary gates \cite{dawson2005solovaykitaevalgorithm,Barenco1995elementary}. \footnote{In practice, modern libraries such as Qiskit or Pennylane can be used to specify a custom unitary gate which is then transpiled into hardware native gates by the software.}

\begin{figure}[h!]
    \centering
    \includegraphics[width=\columnwidth]{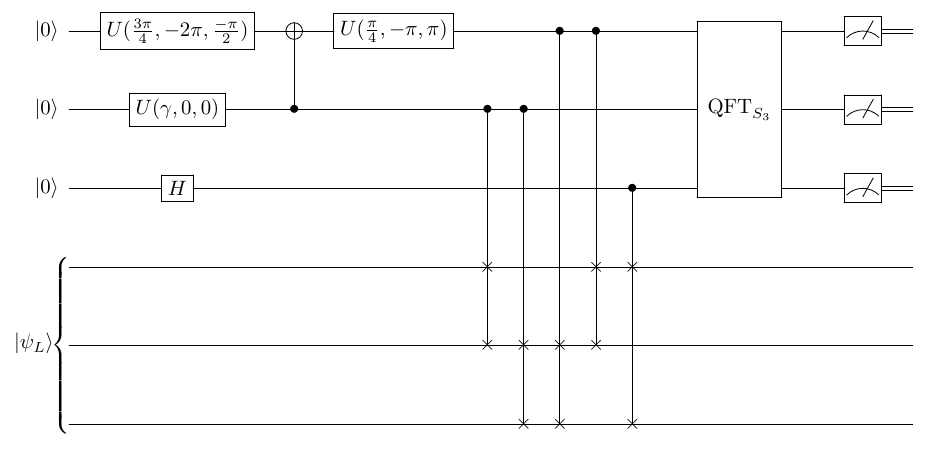}
    \caption{Full $D_3$ (equivalently, $S_3$) circuit according to the error detection method. $\ket{\psi}_L$ is the input logical state. The value of $\gamma$ is given by $\gamma = 2 \cos^{-1} \left ( \frac{{\csc{\frac{\pi}{8}}}\sec{\frac{\pi}{8}}}{2 \sqrt{3}} \right) = 1.23096$. }
    \label{fig:d3_full_circuit}
\end{figure}

Figure~\ref{fig:d3_full_circuit} depicts the full circuit for this example. As previously stated, the form of the group QFT does not depend on the representation on the data register. For example, the representation where $f=\cnot$ and $r= \cnot \swapop$ could result in the same $\operatorname{QFT}_{D_3}$ if the same ancilla mapping is used. 

However, we can easily formulate an example that would give rise to a different implementation. Consider the group element $r^2 =(1\ 2\ 3)^2 = (1\ 3\ 2) $. In the above example, we implemented this element independently, controlled on the ancilla state $\ket{10}_{01}\otimes\ket{+}_2$. Instead, now choose to implement a second cyclic shift of $(1\ 2\ 3)$ alongside the first shift. The ancilla state controlling such an operation would then be $\ket{11}_{01}\otimes\ket{+}_2$. This would permute the rows of $\operatorname{QFT}_{D_3}$ to be
\begin{align*}
 \left( 
\begin{smallmatrix}
 c_{1,1,1}(e) &  c_{1,1,1}(f) & c_{1,1,1}(r) & c_{1,1,1}(rf) & 0 & 0 & c_{1,1,1}(r^2) & c_{1,1,1}(r^2f)  \\
 c_{\text{-1},1,1}(e) & c_{\text{-1},1,1}(f) & c_{\text{-1},1,1}(r) & c_{\text{-1},1,1}(rf) & 0 & 0& c_{\text{-1},1,1}(r^2) & c_{\text{-1},1,1}(r^2f)  \\
 c_{2,1,1}(e) & c_{2,1,1}(f)  & c_{2,1,1}(r) & c_{2,1,1}(rf)   & 0 & 0 & c_{2,1,1}(r^2)  & c_{2,1,1}(r^2f)   \\
 c_{2,1,2}(e)  & c_{2,1,2}(f)  & c_{2,1,2}(r) & c_{2,1,2}(rf) & 0 & 0 & c_{2,1,2}(r^2) & c_{2,1,2}(r^2f)  \\
 0 & 0 & 0 & 0 & 1 & 0 & 0 & 0 \\
 0 & 0 & 0 & 0 & 0 & 1 & 0 & 0 \\
c_{2,2,1}(e)  & c_{2,2,1}(f)  & c_{2,2,1}(r) & c_{2,2,1}(rf) & 0 & 0 & c_{2,2,1}(r^2) & c_{2,2,1}(r^2f) \\
c_{2,2,2}(e)  & c_{2,2,2}(f)  & c_{2,2,2}(r) & c_{2,2,2}(rf) & 0 & 0 & c_{2,2,2}(r^2) & c_{2,2,2} (r^2f)  \\ 
\end{smallmatrix}
\right) ,
\end{align*}
which reflects both that the unused ancilla states are now $\ket{100}$ and $\ket{101}$ and guarantees that the group elements $\ket{r^2}$ and $\ket{r^2f}$ will be mapped to the correct ancillary state. 

\subsection{The $D_{2^{m}}$ QFT}\label{sec:qft-dn}
There is an efficient implementation of the dihedral QFT that takes advantage of the simple semidirect product structure of the dihedral group $D_n$. Indeed, to generate the QFT over the irreps for an arbritrary dihedral group, we borrow the approach from \cite{Alam2022dihedral} which itself builds upon the metacyclic group QFT given in \cite{hoyer1997}. For the sake of completeness, we recount the construction here.

A key insight of this approach relies on a property of the subgroup $\mathbb{Z}_n$ present in the dihedral group $D_n$. In general, if there exists a subgroup $H \subset G$ and a corresponding set of elements $\{g_i\}_{i=1}^m \in G$ such that 
\begin{equation} \label{eq:transversal}
    G = \cup_{i=1}^m g_i H \, ,
\end{equation}
then there is said to exist a left transversal of $H$ in $G$. It is readily apparent that such a left transversal of $\mathbb{Z}_n$ exists in $D_n$ as every group element can be written as $g = f^{\alpha} r^{k}$ with $\alpha \in \mathbb{Z}_2$ and $k \in \mathbb{Z}_n$. This form is reflected in the structure of the ancilla register in the previous section's example, as each ket can be written
\begin{equation}
    \ket{g} = \ket{\alpha} \ket{k} \, ,
\end{equation}
by simply permuting the final register to be the first, without loss of generality. 

This property can be used to find the Fourier transform of a function $f$ over $G$, denoted $F_G$, via its Fourier transform over $H$, denoted $F_H$. Indeed, given the property in \eqref{eq:transversal},
\begin{align}
    \sum_{g \in G}f(g) \rho(g) &= \sum_{i=1}^m\sum_{h \in H} f(g_i h) \rho(g_i h) \, , \\
    &= \sum_{i=1}^m \rho(g_i) \sum_{h \in H} f_i(h) \rho(h) \, ,\\
    &= \sum_{i=1}^m \rho(g_i) \tilde{f}_i(\rho|_H) \, ,
\end{align}
where $f_i(h) \coloneq f(g_i h)$. The above derivation is based on the definition of an $H$-adapted representation provided in \cite[Section~6]{hoyer1997}. 

We can now construct the $\operatorname{QFT}_{D_n}$ of a quantum state via the following procedure
\begin{align}\label{eq:FourierChain}
    \ket{\psi} &= \sum_{g \in G} c(g) \ket{g} \, ,\\
    &= \sum_{\alpha =0}^1\sum_{k \in \mathbb{Z}_n} c(\alpha k)\ket{\alpha}\ket{k} \, , \\
    &\xrightarrow{F_H} \sum_{\alpha =0}^1\ket{\alpha} \sum_{\tilde{k} \in \tilde{\mathbb{Z}_n}} c(\alpha \tilde{k})\ket{\tilde{k}} \, ,\\
    &\xrightarrow{U} \sum_{\tilde{g} \in \tilde{G}} c(\tilde{g}) \ket{\tilde{g}} = \ket{\tilde{\psi}} \, .
\end{align}
In the final step, $U$ is a change of basis unitary defined such that the following diagram commutes 
\begin{center}
        \begin{tikzcd}
        g_i h \arrow{r}{\phi} \arrow{d}{F_H} & g \arrow{d}{F_G} \\
        g_i \tilde{h}  \arrow{r}{U} & \tilde{g}
    \end{tikzcd}
\end{center}
or, equivalently, 
\begin{equation}\label{eq:QFTviasubgroup}
    F_G = U \circ \left ( \mathbb{I}\otimes F_H\right ) \circ \phi^{-1}   \, ,
\end{equation}
where $\phi$ is the bijection matching each $g$ with its unique decomposition $g_ih$. 

For $D_n$ with $n=2^m$, the full circuit is shown in Figure~\ref{fig:dn_qft}, and we now verify that it produces the QFT over the dihedral group.

\begin{figure}[h!]
    \centering
    \includegraphics[width=\columnwidth]{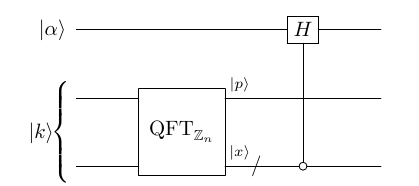}
    \caption{The $\operatorname{QFT}_{D_n}$ circuit where $\ket{k}$ designates the rotational register and $\ket{\alpha}$ denotes the flip register.}
    \label{fig:dn_qft}
\end{figure}

The algorithm begins with the archetypal QFT over $\mathbb{Z}_n$, performing the transformation %unify $\tilde{k}$ versus j notation
\begin{align}
    \ket{\alpha}\ket{k}&\mapsto\frac{1}{\sqrt{n}}\sum_{j=0}^{n-1}\omega^{-jk}\ket{\alpha}\ket{j}\\
    &=\frac{1}{\sqrt{n}}\sum_{p=0}^1\sum_{x=0}^{\frac{n}{2}-1}\omega^{-(p\frac{n}{2}+x)k}\ket{\alpha}\ket{p}\ket{x}.
\end{align}
This is equivalent to
\begin{align}
    \frac{1}{\sqrt{n}}\bigg[\sum_{p=0}^1(-1)^{pk}\ket{\alpha}\ket{p}\ket{0}
    +\sum_{p=0}^1\sum_{x=1}^{\frac{n}{2}-1}\omega^{-(p\frac{n}{2}+x)k}\ket{\alpha}\ket{p}\ket{x}\bigg],
\end{align}
where we have isolated the $x=0$ part of the sum and noted that $\omega^{pk\frac{n}{2}}=(-1)^{pk}$. Finally, a QFT over $\mathbb{Z}_2$ is implemented on the $\ket{\alpha}$ register, controlled off of the all-zeros $\ket{x}$ bitstring. Conveniently, $\operatorname{QFT}_{\mathbb{Z}_2}$ is exactly a single qubit Hadamard gate. This produces the state
\begin{equation}
\begin{aligned}
    \frac{1}{\sqrt{n}}\bigg[\frac{1}{\sqrt{2}}\sum_{p=0}^1(-1)^{pk}&(\ket{0}+(-1)^\alpha\ket{1})\ket{p}\ket{0}\\
    &+\sum_{p=0}^1\sum_{x=1}^{\frac{n}{2}-1}\omega^{-(p\frac{n}{2}+x)k}\ket{\alpha}\ket{p}\ket{x}\bigg].
\end{aligned}
\end{equation}
Pulling out the factor of $\frac{1}{\sqrt{2}}$ and writing out the components of the sum over $p$, we see that this state is equivalent to
\begin{equation}
\begin{aligned}
    \frac{1}{\sqrt{2n}}\bigg[&\ket{0}\ket{0}\ket{0}+(-1)^\alpha\ket{1}\ket{0}\ket{0}\\
    &(-1)^k\ket{0}\ket{1}\ket{0}+(-1)^k(-1)^\alpha\ket{1}\ket{1}\ket{0}\\
    &+\sqrt{2}\sum_{p=0}^1\sum_{x=1}^{\frac{n}{2}-1}\omega^{-(p\frac{n}{2}+x)k}\ket{\alpha}\ket{p}\ket{x}\bigg].
\end{aligned}
\end{equation}
From here, the irreps of $D_n$ are clearly present. Since $n$ is even, there are four 1-dimensional irreps given by $1$, $(-1)^\alpha$, $(-1)^k$, and $(-1)^\alpha(-1)^k$, which show up as the first four coefficients. 

There are $\frac{n}{2}-1$ remaining 2-dimensional irreps, and they appear in the final sum. To see this, first recall that for $\alpha=0$, the representation is diagonal with elements $\omega^{\lambda k}$ and $\omega^{-\lambda k}$, where $\lambda$ indexes the irrep. Similarly, if $\alpha=1$, the representation is antidiagonal with the same elements. When $p=0$, the coefficient in the sum is $\omega^{-kx}$ and when $p=1$, the coefficient is $\omega^{-(\frac{n}{2}+x)k}$. If we substitute $x\mapsto\frac{n}{2}-x$, then this coefficient becomes $\omega^{kx}$. Thus, the matrix elements of the $x$-th representation appear as the coefficients of $\ket{\alpha}\ket{0}\ket{x}$ and $\ket{\alpha}\ket{1}\ket{\frac{n}{2}-x}$. We therefore need to relabel the basis to get the sum in the form shown in \eqref{eq:group-QFT}. We accomplish this by writing $\ket{\lambda,i,j}:=\ket{i\oplus j}\ket{j}\ket{j\frac{n}{2}+(-1)^j\lambda}$.

\subsection{Nonabelian Codes with Constant Depth QFT}

A central structural advantage of the stabilizer-code formalism is that the QFT over $\mathbb{Z}_2^{\oplus n}$ required for syndrome extraction is a constant-depth unitary. Indeed, it is simply the tensor product of Hadamard gates,
\begin{equation}
    \mathrm{QFT}_{\mathbb{Z}_2^{\oplus n}}
    = H^{\otimes n}.
\end{equation}
For nonabelian groups, this feature generally fails. As demonstrated in Section~\ref{sec:qft-dn}, the QFT over the dihedral group $D_n$ has logarithmic circuit depth in the group index~$n$. Since $n$ is determined by the number of qubits, this translates into a linear depth in the number of physical qubits when the logical subspace is encoded using the representation constructed in Section~\ref{sec:one-logical}. This prompts the natural question:

\begin{quote}
\emph{Are there nonabelian codes whose syndrome-extraction procedure still
admits a constant-depth QFT, analogous to the stabilizer case?}
\end{quote}

A simple but powerful construction is obtained by embedding the desired nonabelian symmetry into a product group for which the QFT factors. Specifically, let
\begin{equation}
    G \;=\; K \times \mathbb{Z}_2^{\oplus n},
\end{equation}
where $K$ is any fixed nonabelian finite group.  The irreducible representations of the product group are tensor products of the irreps of the factors, and the QFT respects this decomposition:
\begin{equation}
    \mathrm{QFT}_{K \times \mathbb{Z}_2^{\oplus n}}
    \;=\;
    \mathrm{QFT}_K \,\otimes\, \mathrm{QFT}_{\mathbb{Z}_2^{\oplus n}}.
\end{equation}
Since $K$ is fixed, the circuit implementing $\mathrm{QFT}_K$ has constant depth, depending only on the choice of $K$.  Meanwhile, $\mathrm{QFT}_{\mathbb{Z}_2^{\oplus n}} = H^{\otimes n}$ is also constant depth. Hence,
\begin{equation}
    \mathrm{depth}(\mathrm{QFT}_{K \times \mathbb{Z}_2^{\oplus n}})
    \;=\;
    O(1),
\end{equation}
independent of the code size.

The corresponding code space is the $G$-symmetric subspace $\mathrm{Sym}_G$ of the representation $W=W_1 \otimes W_2$ of $G=K\times \mathbb{Z}_2^{\oplus n}$, where $W_1$ and $W_2$ are chosen representations of $K$ and $\mathbb{Z}_2^{\oplus n}$, respectively. Note that if $\rho$ is an irrep of $K$ and $\sigma$ is an irrep of $\mathbb{Z}_2^{\oplus n}$, then $\rho\otimes\sigma$ is an irrep of $K\times\mathbb{Z}_2^{\oplus n}$. This fact, along with the resulting fact that the character splits into a product $\chi_{\rho\otimes\sigma}(k,\Vec{n})=\chi_\rho(k)\chi_\sigma(\Vec{n})$, can be used to show that the projection onto an isotypic component splits as a tensor product. Indeed, if the irreps of $K$ are labeled by $\lambda$ and the irreps of $\mathbb{Z}_2^{\oplus n}$ are labeled by $\mu$, then
\begin{equation}
    \Pi^{K\times\mathbb{Z}_2^{\oplus n}}_{\lambda,\mu}=\Pi_{\lambda}^K \otimes \Pi_{\mu}^{\mathbb{Z}_2^{\oplus n}},
\end{equation}
where $\Pi_{\lambda}^K$ projects onto the $\lambda$-isotypic component of the $K$-representation $W_1$ and $\Pi_{\mu}^{\mathbb{Z}_2^{\oplus n}}$ is the Pauli-type projector arising from the abelian factor with representation $W_2$. The group in the superscript has been added to help distinguish between the projectors.

This construction therefore yields a broad class of nonabelian symmetry-protected codes with stabilizer-like measurement complexity: nontrivial representation-theoretic structure is introduced in the $K$ sector, while the $\mathbb{Z}_2^{\oplus n}$ factor guarantees constant-depth syndrome extraction. Thus, although general nonabelian QFTs require growing circuit depth, direct-product constructions show that nonabelian symmetry is compatible with constant-depth QFT-based syndrome extraction. Any fixed nonabelian group may be incorporated into a scalable code in this manner, allowing one to preserve the fault-tolerant advantages of stabilizer codes while extending the symmetry structure beyond the abelian setting.

\section{Conclusion}\label{sec:conclusion}
We have presented a unifying representation-theoretic framework for quantum error correction in which the code space is defined as the symmetric subspace of a finite group representation. Within this setting, stabilizer codes arise as a special case corresponding to abelian subgroups of the Pauli group, and their standard syndrome extraction procedure is reinterpreted as isotypic projection onto irreducible components. Our construction generalizes naturally to qudit stabilizer codes and to codes built from nonabelian group symmetries, thereby broadening the scope of symmetry-protected quantum error correction.

The primary advantage of this framework is its flexibility; the same formalism can accommodate abelian and nonabelian structures, discrete and mixed-dimension qudits, and even symmetry actions not naturally described by Pauli operators. Moreover, the isotypic syndrome extraction procedure discretizes errors in a manner analogous to stabilizer codes but with respect to more general symmetry groups. In this sense, stabilizer codes are not isolated phenomena but rather the simplest instance of a more general symmetry-aware principle.

There remain several challenges and open directions. Implementing efficient QFTs for large or nonabelian groups may be resource intensive, and a deeper study of fault-tolerant realizations within this framework is warranted. The classification of logical operations is also more subtle than in the stabilizer formalism, particularly due to the structure of the normalizer and its action on isotypic projectors. Furthermore, while we have highlighted conceptual connections to subsystem codes and decoherence-free subspaces, the relationship between these approaches deserves further exploration.

Looking forward, it will be valuable to investigate the performance of symmetry-protected codes under realistic hardware noise models, benchmark their threshold properties against leading architectures such as the surface code, and explore approximate group Fourier transforms to reduce overhead. Beyond fault tolerance, the present framework suggests potential applications in quantum simulation, where group symmetries play a central role, and in tailoring codes to the natural noise of specific physical platforms. We view these directions as promising steps toward developing a comprehensive, symmetry-resolved theory of quantum error correction.
% \begin{acknowledgments}
% The authors thank Ada H. Bradshaw for her expert advice and keyboard smashes.
% \end{acknowledgments}

\bibliographystyle{apsrev4-2}
\bibliography{refs}

\end{document}